\documentclass[psamsfonts,fceqn,leqno]{amsart}
  \usepackage{amsthm}
  \usepackage{amsmath}
  \usepackage{amssymb}
  \usepackage{mathrsfs}
  
  \newtheorem{theorem}{Theorem}[section]
  \newtheorem{corollary}[theorem]{Corollary}
  \newtheorem{proposition}[theorem]{Proposition}
  \newtheorem{lemma}[theorem]{Lemma}

  \theoremstyle{definition}
  \newtheorem{definition}[theorem]{Definition}
  
  \newtheorem{example}[theorem]{Example}

  \numberwithin{equation}{section}

  \title[Ex-post Core, Fine Core and REE Allocations]{Ex-post Core, 
  Fine Core and Rational Expectations Equilibrium Allocations}
  \author[A. Bhowmik]{Anuj Bhowmik}
  \address{Department of Economics, Shiv Nadar University, NH91, Tehsil Dadri, Gautam Buddha 
  Nagar, Uttar Pradesh 201314, India}
  \email{anujbhowmik09@gmail.com}

  \author[J. Cao]{Jiling Cao}
  \address{Department of Mathematical Sciences, School of Engineering, Computer and Mathematical 
  Sciences, Auckland University of Technology, Private Bag 92006, Auckland 1142, New Zealand}
  \email{jiling.cao@aut.ac.nz}

  \thanks{\hspace{-1.66em} \emph{JEL classification:} D41; D43; D51; D82.}

  \thanks{\noindent The second-named author thanks the support of the National Natural Science 
  Foundation of China, grant No. 11571158, and the paper was partially written when he visited 
  Minnan Normal University in April 2016 as Min Jiang Scholar Guest Professor.}

  \date{}

\begin{document}

  \maketitle

  \begin{abstract}
  This paper investigates the ex-post core and its relationships to the fine core and the set of
  rational expectations equilibrium allocations in an oligopolistic economy with asymmetric 
  information, in which the set of agents consists of some large agents and a continuum of small 
  agents and the space of states of nature is a general probability space. We show that under 
  appropriate assumptions, the ex-post core is not empty and contains the set of rational 
  expectations equilibrium allocations. We provide an example of a pure exchange continuum economy 
  with asymmetric information and infinitely many states of nature, in which the ex-post core does 
  not coincide with the set of rational expectations equilibrium allocations. We also show that 
  when our economic model contains either no large agents or at least two large agents with the 
  same characteristics, the fine core is contained in the ex-post core. 
  \end{abstract}

  \section{Introduction} \label{sec:intro}
  
  In general equilibrium theory, the core and competitive equilibrium are two important solution 
  concepts. For an exchange economy with complete information, the core and its relationship to 
  the set of competitive allocations have been studied intensively in the literature (for a 
  comprehensive survey, refer to \cite{Anderson:92}). In the past few decades, several alternative 
  cooperative and non-cooperative equilibrium concepts have been proposed, in the context of 
  asymmetric information economies. The core of an economy with asymmetric information was first 
  considered by Wilson \cite{Wilson:78}, where the concepts of coarse and fine core were proposed. 
  The fine core presumes that agents can share their information when they form a coalition and an 
  allocation is not in the fine core, if a coalition has some distribution of the total endowments 
  of its members which gives to all of its members a better pay-off in an event which the coalition 
  can jointly discern. In \cite{Yannelis:91}, Yannelis introduced the concept of private core, 
  which is an analogue concept to the core for an economy with complete (and symmetric) information, 
  and proved that under appropriate assumptions, the private core is always non-empty. 
  In the definition of the private core, when a coalition blocks an allocation, each member in the 
  coalition uses only his own private information. Furthermore, Einy et al. 
  \cite{Einy-Moreno-Shitovitz:00b, Einy-Moreno-Shitovitz:00} studied the notion of ex-post core, 
  in the sense that an ex-post core allocation cannot be ex-post blocked by any coalition. On the 
  other hand, Radner \cite{Radner:79} introduced the notion of a (Bayesian) rational expectations 
  equilibrium by imposing the Bayesian (subjective expected utility) decision doctrine, in order 
  to capture the information revealed by the market clearing price. The fact that a Bayesian 
  rational expectations equilibrium does not exist universally motivates de Castro et al. 
  \cite{de Castro-Pesce-Yannelis:11} to introduce the concept of a maxinmin rational expectations 
  equilibrium, by replacing the Bayesian decision-making approach of Radner with the maximin 
  expected utility. A good survey article for the equilibrium concepts in asymmetric (or 
  differential) information economies is \cite{Glycopantis:2005}.
  
  \medskip
  For economies with complete information, Aumann \cite{Aumann:64} proved that competitive and 
  core allocations coincide, provided that there is a continuum of traders. The existence of such 
  allocations was studied by Aumann \cite{Aumann:66} and Hildenbrand \cite{Hildenbrand:70}. 
  Extensions of these results to economies with asymmetric information were made by Einy et 
  al. \cite{Einy-Moreno-Shitovitz:00b, Einy-Moreno-Shitovitz:01}. In \cite{Einy-Moreno-Shitovitz:00b}, 
  Einy et al. first established some representation results on the ex-post core and the set of 
  rational expectations equilibrium allocations. Then, these representations results together with 
  Aumann's Core Equivalence Theorem enabled them to show that if the economy is atomless and the 
  utility function of each agent is measurable with respect to his information, then the set of 
  rational expectations equilibrium allocations coincides with the ex-post core. In 
  \cite{Einy-Moreno-Shitovitz:01}, Einy et al. showed that, if an economy is irreducible, then a 
  competitive (or Walrasian expectations) equilibrium exists and, moreover, the set of competitive 
  equilibrium allocations coincides with the private core. However, to obtain these results, they 
  allow for free disposal on the feasibility (market clearing) constraints. This was motivated 
  by an example \cite{Einy-Moreno-Shitovitz:01} of an economy with asymmetric information which 
  has a competitive equilibrium with free disposal, but if the feasibility constraints are 
  imposed with an equality, then the economy does not have a competitive equilibrium where prices 
  of all contingent contracts for future delivery are non-negative. In a few years later, Angeloni 
  and Martins-da-Rocha \cite{angeloni:09} proved that the results in \cite{Einy-Moreno-Shitovitz:01} 
  are still valid without free-disposal. 
  
  \medskip
  In the past few years, techniques have been developed by Bhowmik et al. in \cite{BCY:14} to 
  investigate the existence of rational expectations equilibrium in a general model of pure 
  exchange economies. Moreover, Bhowmik and Cao \cite{Bhowmik-Cao:16} established a representation 
  result for rational expectations equilibrium allocations in terms of the state-wise Walrasian 
  allocations. As a rational expectation equilibrium allocation is an interim solution concept 
  and it takes into account the information of all other agents through market price, Bhowmik 
  and Cao \cite{Bhowmik-Cao:16} showed their result by assuming that each agent knows his initial 
  endowment and utility. Such assumptions lead to a fact that the information revealed by prices 
  play no role and thus, the Bayesian (maximin) rational expectation equilibrium allocations 
  becomes almost the same as the state-wise Walrasian allocations\footnote{But they are not the 
  same as both $T$ and $\Omega$ are infinite, see Example \ref{exam:sinclusion} in this paper.}. 
  
  \medskip
  Our aim of this paper is to apply the results and techniques developed in \cite{BCY:14, 
  Bhowmik-Cao:16} to the study of the ex-post core and its relationships to the fine core and the 
  set of rational expectations equilibrium allocations. We consider an oligopolistic economy with 
  asymmetric information in which the set of agents consists of some large agents and a continuum 
  of small agents. The uncertainty is model by a general probability space of states of nature in 
  which each agent is characterized by a state-dependent utility function, a random initial 
  endowment, an information partition and a prior belief. Firstly, we establish a result on the 
  existence and characterization of the ex-post core, which can be regarded as an extension of the 
  corresponding result in \cite{Einy-Moreno-Shitovitz:00b} to a framework with infinitely many states 
  of the nature. The proof of this result relies on the measurability of Walrasian equilibrium 
  correspondence with respect to the information structure in the economy (see Theorem 
  \ref{thm:mainexpost}). In the presence of the result in \cite{Bhowmik-Cao:16} and Aumann's Core 
  Equivalence Theorem, we conclude that Bayesian (maximin) rational expectation equilibrium 
  allocations are contained the ex post core. This is a version of the first fundamental theorem 
  of social welfare for large economies with asymmetric information. However, contrary to the 
  equivalence result for finitely many states of nature in \cite{Einy-Moreno-Shitovitz:00b}, we 
  provide an example of a continuum economy with asymmetric information and infinitely many states 
  of nature, in which the ex-post core strictly contains all rational expectations equilibrium 
  allocations. This means that the core-Walras equivalence can fail in a continuum economy with 
  asymmetric information when it has infinitely many states of nature. Secondly, we show that 
  under appropriate assumptions and the assumption that there are only finitely many different 
  information structures and all information is the joint information of agents, the fine core is 
  contained in the ex-post core. This extends the corresponding result in 
  \cite{Einy-Moreno-Shitovitz:00}. To obtain this result, following \cite{Greenberg-Shitovitz:86}, 
  we first associated an atomless economy with our oligopolistic economy so that all large agents 
  are broken into a continuum of small agents with similar characteristics. The idea of the proof 
  is as follows: if an allocation is not in the ex-post core of our original economy, it must
  not be a core allocation in some complete information economy and so in the corresponding 
  complete information atomless economy. Vind's theorem (see \cite{Vind:72}) implies that an 
  arbitrary large coalition can be chosen so that it discerns any state of nature. With the help 
  of some other techniques, we are able to show that the allocation is blocked by a coalition that 
  jointly has full information in our original economy and thus, it is not in the fine core.      
  
  \medskip
  The structure of the paper is as follows. Section \ref{sec: modelpluscore} presents the 
  theoretical framework and outlines the basic model. We also study several correspondences 
  associated with our basic model. These correspondences form the major part of our tool kits. 
  Section \ref{sec:ree_vs_epc} investigates a representation of the ex-post core and its 
  relationship with the set of rational expectations equilibrium allocations. Section 
  \ref{sec:FineCore} studies the relationship between the ex-post core and the fine core.
  Finally, we provide some concluding remarks in Section \ref{sec:conclusion}.
  
  \section{The Model and Associated Correspondences}
  \label{sec: modelpluscore}

  In this section, we describe a basic model of a pure exchange mixed economy with asymmetric 
  information.

  \subsection{The model} \label{subsec: model}

  We consider a pure exchange economy $\mathscr E$ with asymmetric information. The exogenous 
  uncertainty is described by a probability space $(\Omega, \mathscr{F}, \mathbb P)$, where 
  $\Omega$ is a set denoting all possible states of nature, the $\sigma$-algebra $\mathscr{F}$ 
  denotes possible events, and $\mathbb P$ is a complete probability measure. The space of agents 
  is a measure space $(T,\Sigma,\mu)$ with a complete, finite and positive measure $\mu$, where 
  $T$ is the set of agents, $\Sigma$ is the $\sigma$-algebra of measurable subsets of $T$ whose 
  economic weights on the market are given by $\mu$. Since $\mu(T)<\infty$, a classical result 
  in measure theory claims that $T$ can be decomposed into the union of two parts: one is atomelss 
  and the other contains at most countably many atoms, that is, $T= T_0 \cup T_1$, where $T_0$ is 
  the atomless part and $T_1$ is the union of at most countably many $\mu$-atoms, refer to 
  \cite[p.155]{Nielsen:97}. Let $\mathscr A =\{A_n: n\ge 1\}$ be the family of all atoms in 
  $T_1$, i.e., $T_1 =\bigcup_{n\ge 1} A_n$. Agents in $T_0$ are called ``\emph{small agents}", 
  who are un-influential agents (the price takers). According to a standard interpretation, we 
  can think that each $A_n$ arises from a group of small identical agents that decide to join 
  and to act on the market only together. As consequence of such agreements, no proper 
  subcoalitions of the group are possible and then the group is identified with an atom of $\mu$. 
  Agents in $T_1$ are called ``\emph{large agents}", who are influential ones (the oligopolies). 
  With an abuse of notation, we shall identify $T_1$ with $\mathscr A$, i.e., $T_1 = {\mathscr A}$. 
  The commodity space is the $\ell$-dimensional Euclidean space $\mathbb R^\ell$. For $\lambda >0$, 
  $B(0, \lambda)$ denotes the ball in $\mathbb R^\ell$ centred at $0$ with radius $\lambda$. The 
  partial order on $\mathbb R^\ell$ is denoted by $\leq$. More precisely, for any two vectors 
  $x =(x_1,..., x_\ell)$ and $y=(y_1,...,y_\ell)$ in $\mathbb R^\ell$, we write 
  $x\le y$ (or $y \ge x$) if $x_k \le y_k$ for all $1\le k \le \ell$. Furthermore, we write $x<y$ 
  (or $y>x$) when $x\le y$ and $x\ne y$, and $x\ll y$ (or $y\gg x$) when $x_k < y_k$ for all $1\le 
  k \le \ell$. Let $\mathbb R^\ell_+= \{x\in \mathbb R^\ell: x \ge 0\}$, and let $\mathbb 
  R^\ell_{++}=\{x\in \mathbb R^\ell_+: x\gg 0\}$. In each state, the \emph{consumption set} for 
  every agent $t\in T$ is $\mathbb R^\ell_+$.  Each agent $t\in T$ is characterized by a quadruple 
  $({\mathscr F}_t, U(t,\cdot,\cdot), a(t, \cdot), \mathbb P_t)$, where
  \begin{enumerate}
  \item[(i)] ${\mathscr F}_t$ is the $\sigma$-algebra generated by a measurable partition 
  $\Pi_t$ of $\Omega$ representing the \emph{private information} of agent $t$,
  \item[(ii)] $U(t,\cdot,\cdot): \Omega \times \mathbb R^\ell_+ \to\mathbb R$ is the 
  \emph{state-dependent utility function} of agent $t$,
  \item[(iii)] $a(t,\cdot): \Omega\rightarrow \mathbb R^\ell_+$ is the \emph{state-dependent 
  initial endowment} of agent $t$, and
  \item[(iv)] $\mathbb P_t$ is a probability measure on $\mathscr F$, giving the \emph{prior 
  belief} of agent $t$.
  \end{enumerate}
  The quadruple $({\mathscr F}_t, U(t,\cdot,\cdot), a(t, \cdot), \mathbb P_t)$ is sometimes 
  known as \emph{characteristics} of agent $t$. Two agents are said to be the \emph{same 
  type} if they have the same characteristics. Formally, the economy ${\mathscr E}$ can be 
  expressed by
  \[
  {\mathscr E} = \{(\Omega, \mathscr{F}, \mathbb P);\ (T,\Sigma,
  \mu);\ \mathbb R^\ell_+;\ ({\mathscr F}_t, U(t,\cdot,\cdot),
  a(t, \cdot), \mathbb P_t)_{t\in T}\}.
  \]
  In the complete information Arrow-Debreu-McKenzie model, prices are vectors in $\mathbb 
  R^\ell_+ \setminus \{ 0\}$. Following the standard treatment in the literature (e.g., see 
  \cite{Aumann:66}), price vectors are normalized so that their sum is 1. 
  
  \medskip
  In this paper, we use the symbol $\Delta$ to denote the simplex of normalized price 
  vectors, i.e.,
  \[
  \Delta= \left\{p \in \mathbb R^\ell_{+}: \sum_{h= 1}^\ell 
  p^h =1\right\}.
  \]
  Put  $\Delta_+= \Delta\cap \mathbb R^\ell_{++}$. Throughout the paper, $\Delta$ and 
  $\Delta_+$ are equipped with the relative Euclidean topology. A \emph{price system} of 
  $\mathscr E$ is an $\mathscr F$-measurable function $\pi: \Omega \to \Delta$, where 
  $\Delta$ is equipped with the Borel structure $\mathscr B(\Delta)$ generated by the 
  relative Euclidean topology. 
  
  \medskip
  Let $\sigma(\pi)$ be the smallest $\sigma$-algebra contained in $\mathscr F$ and 
  generated by a price system $\pi$. Intuitively, $\sigma(\pi)$ represents the information
  revealed by $\pi$. The combination of agent $t$'s private information $\mathscr F_t$ 
  and the information revealed by the price system $\pi$ is given by the smallest 
  $\sigma$-algebra $\mathscr G_t$ that contains both $\mathscr F_t$ and $\sigma(\pi)$. 
  Formally, $\mathscr G_t =\mathscr F_t \vee \sigma(\pi)$. For any $\omega \in \Omega$, 
  let $\mathscr G_t(\omega)$ denote the smallest element of $\mathscr G_t$ that contains 
  $\omega$.

  \medskip
  As interpreted in \cite{de Castro-Pesce-Yannelis:11}, the economy $\mathscr E$ extends 
  over three time periods: ex ante ($\tau=0$), interim ($\tau=1$) and ex post ($\tau=2$).
  At $\tau=0$, the state space, the partitions, the structure of the economy and the 
  price functional $\pi: \Omega \to \Delta$ are common knowledge. This stage does not 
  play any role in our analysis and it is assumed just for a matter of clarity. At 
  $\tau =1$, each individual learns his private information and the prevailing prices 
  $\pi(\omega)$, and thus learns $\mathscr G_t$. With these in his mind, the agent
  plans how much he will consume $x(\omega)$. However, his actual consumption may be
  contingent to the final state of the nature, which is not yet known by him. The 
  individual agent only knows that one of the states $\omega' \in \mathscr G_t(\omega)$ 
  will be realized. Therefore, he needs to make sure that he will be able to pay his 
  consumption plan $x(\omega')$ for all $\omega'\in \mathscr G_t(\omega)$. At $\tau=2$, 
  each individual agent $t\in T$ receives and consumes his entitlement $f_t(\omega)$.
  
  \medskip
  Recall that a function $u: \mathbb R^\ell_+ \to \mathbb R$ is \emph{strictly increasing} 
  if $u(x) < u(y)$ for any $x, y \in \mathbb R^\ell_+$ with $x< y$, and it is \emph{quasi 
  concave} if
  \[
  u(\alpha x + (1-\alpha)y) \ge \min \{u(x), u(y)\}
  \]
  for any $x, y \in \mathbb R^\ell_+$ with $x \ne y$ and any $0< \alpha < 1$. In ``$\ge$" 
  in the above inequality is replaced with ``$>$", then $u$ is called \emph{strictly quasi 
  concave}.

  \medskip
  Throughout the paper, the following standard assumptions will be used. These assumptions 
  are similar to those in \cite{BCY:14, Bhowmik-Cao:16}.

  \medskip
  \noindent
  ({A}$_1$) The initial endowment function $a:(T,\Sigma,\mu)\times (\Omega, \mathscr F, 
  \mathbb P) \rightarrow \mathbb R^\ell_+$ is $\Sigma \otimes {\mathscr F}$-measurable 
  such that $a(\cdot, \omega)$ is Bochner integrable and $\displaystyle \int_T a(\cdot, 
  \omega)d\mu \gg 0$ for each $\omega \in\Omega$.
  
  \medskip
  \noindent
  (${\rm A}'_1$) The initial endowment function $a:(T,\Sigma,\mu)\times (\Omega, \mathscr 
  F, \mathbb P) \rightarrow \mathbb R^\ell_+$ is $\Sigma \otimes {\mathscr F}$-measurable 
  such that $a(\cdot, \omega)$ is Bochner integrable and $a(\cdot, \omega) \gg 0$ $\mu$-a.e. 
  on $T$ for each $\omega \in\Omega$.

  \medskip
  \noindent
  ({A}$_2$) $U(\cdot,\cdot, x): (T,\Sigma,\mu) \times (\Omega, \mathscr F, \mathbb P) \to 
  \mathbb R$ is $\Sigma \otimes {\mathscr F}$-measurable for all $x \in \mathbb R^\ell_+$.

  \medskip
  \noindent
  ({A}$_3$) For each $(t, \omega) \in T\times \Omega$, $U(t, \omega, \cdot): \mathbb 
  R^\ell_+ \to \mathbb R$ is continuous and strictly increasing.
   
  \medskip
  \noindent
  ({A}$_4$)  For each $(t,\omega)\in T\times \Omega$, $U(t,\omega,\cdot)$
  is strictly quasi-concave.
  
  \medskip
  \noindent
  (${\rm A}'_4$) For each $(t,\omega)\in T_1\times \Omega$, $U(t, \omega, \cdot)$ 
  is quasi-concave.
  
  \medskip
  Here, we would like to add some comments on these assumptions. Note that the condition 
  ``$\displaystyle \int_T a(\cdot, \omega)d\mu \gg 0$ for each $\omega \in\Omega$" 
  in ({A}$_1$) or ``$a(\cdot, \omega) \gg 0$ $\mu$-a.e. on $T$ for each $\omega \in\Omega$" 
  in (${\rm A}'_1$), which implies that no commodity is totally absent from the market, 
  has been commonly used for results on the existence of an equilibrium, for instance, see 
  \cite{Aumann:66, BCY:14, Bhowmik-Cao:16, Einy-Moreno-Shitovitz:00b, Einy-Moreno-Shitovitz:00}. 
  The joint measurability of the initial endowment $a$ in ({A}$_1$) and (${\rm A}'_1$) has 
  been used in \cite{BCY:14, Bhowmik-Cao:16} for general models of asymmetric information 
  economies with infinitely many states of nature. The assumption (${\rm A}'_1$) is 
  stronger than ({A}$_1$) and is used in \cite{Bhowmik-Cao:paper2, Bhowmik-Cao:paper3}. 
  Assumption ({A}$_2$) is equivalent to the measurability condition used in \cite{Aumann:64,
  Aumann:66}. Since then, it has been widely used in the literature, see \cite{BCY:14,
  Bhowmik-Cao:16, Einy-Moreno-Shitovitz:00, Einy-Moreno-Shitovitz:00b, Einy-Moreno-Shitovitz:01}.
  Although ({A}$_1$) and ({A}$_2$) are not used in Einy et al. \cite{Einy-Moreno-Shitovitz:00}, 
  $U(t, \cdot, x)$ and $a(t,\cdot)$ are required to be $\mathscr F$-measurable for all $(t,x)
  \in T \times \mathbb R^\ell_+$. Finally, ({A}$_3$), ({A}$_4$) and (${\rm A}'_4$) impose 
  properties on the agents' utility functions. These assumptions have been quite commonly
  used in the literature. 

  \medskip
  A member $S$ of $\Sigma$ with $\mu(S)> 0$ is called a \emph{coalition} of $\mathscr E$. Let 
  $L_1(\mu, \mathbb R^\ell)$ denote the set of all equivalent classes of Bochner integrable 
  functions from $T$ into $\mathbb R^\ell$. An \emph{assignment} in $\mathscr E$ is a function 
  $f: T \times \Omega\rightarrow\mathbb R^\ell_+$ such that for every $\omega \in \Omega$, 
  $f(\cdot, \omega) \in L_1(\mu, \mathbb R^\ell)$, and for every $t\in T$, $f(t,\cdot)$ 
  is $\mathscr F$-measurable. If an assignment $f$ is also \emph{feasible}, i.e., for 
  every $\omega \in \Omega$,
  \[
  \int_T f(\cdot, \omega) d\mu= \int_T a(\cdot, \omega)d\mu,
  \]
  then it is called an \emph{allocation}. Note that under ({A}$_1$), the initial 
  endowment $a$ is an allocation in $\mathscr E$.

  \subsection{Correspondences Associated with $\mathscr E$} \label{sec:associated}
  
  Following \cite{BCY:14, Bhowmik-Cao:16}, we define a function $\delta:\Delta_+\to\mathbb R_{++}$ 
  such that for each $p= (p^1,\cdots, p^\ell)\in \Delta_+$,
  \[
  \delta(p) = \min\left\{p^h: 1\le h \le\ell \right\}.
  \] 
  For any $(t, \omega, p)\in T\times \Omega\times \Delta_+$, let
  \[
  \gamma(t, \omega, p)= \frac{1}{\delta(p)}\sum_{h=1}^\ell a^h(t,
  \omega), \hspace{1em} \mbox{and} \hspace{1em} b(t, \omega,p)= \gamma(t, \omega, p){\bf 1},
  \]
  where ${\bf 1}=(1,\cdots,1)\in \mathbb{R^\ell}$. Define the correspondence $X: T\times \Omega\times 
  \Delta_+ \rightrightarrows \mathbb R^\ell_+$ by 
  \[
  X(t,\omega,p)=\{x\in \mathbb R^\ell_+: x\le b(t,\omega,p)\}
  \]
  for all $(t, \omega, p)\in T\times \Omega\times \Delta_+$. The \emph{budget correspondence} 
  $B: T\times \Omega\times \Delta \rightrightarrows \mathbb R^\ell_+$ is defined by
  \[
  B(t, \omega, p)= \left\{x\in \mathbb
  R^\ell_+: \langle p, x\rangle\le \langle p, a(t,\omega)
  \rangle\right\}
  \]
  for all $(t,\omega,p) \in T \times \Omega \times \Delta$. Note that $X$ and $B$ are non-empty, 
  closed- and convex-valued such that $B(t, \omega, p)\subseteq X(t, \omega, p)$ for all 
  $(t, \omega, p)\in T\times \Omega\times \Delta_+$. Furthermore, the compactness of 
  $X(t, \omega, p)$ implies that $B(t,\omega,p)$ is compact for every $(t, \omega, p)\in T\times 
  \Omega\times \Delta_+$.
  
  \medskip
  Following \cite{Aliprantis-Border:06}, we say that a correspondence $F: (T, \Sigma, \mu)\rightrightarrows 
  \mathbb R^\ell$ is \emph{weakly $\Sigma$-measurable} if 
  \[
  F^{-1}(V) =\{ t \in T: F(t) \cap V \ne \emptyset\} \in \Sigma
  \]
  for all open subset $V$ of $\mathbb R^\ell$. Wherever no confusion arises in the sequel, we shall omit 
  $\Sigma$ in the definition of a weakly $\Sigma$-measurable correspondence. A function $f: (T, \Sigma, \mu)
  \rightarrow \mathbb R^\ell$ is called a \emph{measurable selection of $F$} if $f$ is $\Sigma$-measurable 
  and $f(t) \in F(t)$ $\mu$-a.e..
  
  \begin{lemma}[\cite{Aubin:90}] \label{lem:measurable}
  Let $F: (T, \Sigma, \mu)\rightrightarrows \mathbb R^\ell$ be a correspondence. Then the following 
  statements are equivalent:
  \begin{itemize}
  \item[(i)] $F$ is weakly $\Sigma$-measurable.
  \item[(ii)] $F$ has a measurable graph, that is, ${\rm Gr}_F \in \Sigma \otimes {\mathscr B}
  ({\mathbb R}^\ell)$.
  \item[(iii)] For every $x \in {\mathbb R}^\ell$, ${\rm dist}(x, F(\cdot)): T \to \mathbb R_+$ is 
  $\Sigma$-measurable.
  \end{itemize}
  \end{lemma}  
  
  The following proposition is similar to \cite[Proposition 4.1]{BCY:14} and is a special case of 
  \cite[Lemma 2]{Bhowmik-Cao:16}.
   
  \begin{proposition}\label{prop:budcaratheodory}
  Assume that an economy $\mathscr E$ satisfies \emph{({A$_1$})}. Then $B$ is weakly 
  $\Sigma\otimes \mathscr F\otimes \mathscr B(\Delta)$-measurable and $X$ is weakly 
  $\Sigma\otimes \mathscr F\otimes \mathscr B(\Delta_+)$-measurable.
  \end{proposition}
  
  Define the correspondences $C: T\times\Omega\times\Delta\rightrightarrows \mathbb R^\ell_+$ and  $C^{X}: 
  T\times\Omega\times\Delta_+\rightrightarrows \mathbb R^\ell_+$by
  \[
  C(t, \omega, p)= \left\{ x \in \mathbb R^\ell_+: U(t,\omega, x) \ge U(t, \omega, y) \mbox{ for all } 
  y\in B(t, \omega, p)\right\}
  \]
  and
  \[
  C^X(t, \omega, p)= C(t, \omega, p)\cap X(t, \omega, p).
  \]
  By ({A}$_3$), for every $x\in C(t,\omega,p)$ and $(t,\omega,p) \in T\times\Omega \times\Delta$, 
  $\langle p, x\rangle\ge \langle p, a(t, \omega)\rangle$. Furthermore, it is easy to see that 
  \[
  B(t, \omega, p)\cap C(t, \omega, p)= B(t, \omega, p)\cap C^X (t, \omega, p)
  \]
  for all $(t, \omega, p)\in T\times \Omega \times \Delta_+$. Note that under ({A}$_3$), $U(t, \omega, 
  \cdot)$ is continuous on the non-empty compact set $B(t,\omega, p)$ for all $(t, \omega, p)\in T\times 
  \Omega \times \Delta_+$. Thus, one has
  \[
  B(t, \omega, p)\cap C(t, \omega, p) \ne \emptyset
  \]
  for all $(t, \omega, p)\in T\times \Omega \times \Delta_+$.

  \medskip 
  The following proposition is similar to \cite[Proposition 4.2]{BCY:14}.
  
  \begin{proposition}\label{prop:fmeasure}
  Assume that an economy $\mathscr E$ satisfies \emph{({A$_1$})-({A$_3$})}. Then $C^X$ is weakly
  $\Sigma\otimes \mathscr F\otimes \mathscr B(\Delta_+)$-measurable.
  \end{proposition}
  
  \begin{proof}
  By Proposition \ref{prop:budcaratheodory}, $B$ is weakly $\Sigma\otimes \mathscr F \otimes \mathscr B
  (\Delta)$-measurable. Thus, by \cite[Corollary 18.14]{Aliprantis-Border:06}, there exists a sequence 
  $\{f_n: n \ge 1\}$ of $\Sigma\otimes \mathscr F \otimes \mathscr B (\Delta)$-measurable functions 
  from $T\times \Omega \times \Delta$ to ${\mathbb R}^\ell_+$ such that
  \[
  B(t, \omega,p)= \overline{\{f_n(t,\omega,p): n \ge 1\}}
  \]
  for all $(t,\omega, p)\in T\times \Omega\times \Delta$. 
  
  For each $n \ge 1$, define $C_n: T\times \Omega \times \Delta \rightrightarrows \mathbb R^\ell_+$ by 
  letting
  \[
  C_n(t,\omega,p)= \left\{x\in \mathbb R^\ell_+: U(t, \omega, x)
  \ge U(t, \omega, f_n(t,\omega,p))\right\},
  \]
  and
  $\xi_n:T\times \Omega\times \Delta \times \mathbb R^\ell_+\to \mathbb R$ by letting
  \[
  \xi_n(t, \omega, p, x)=U(t, \omega, x)- U(t, \omega, f_n(t,\omega,p)).
  \]
  Note that $\xi_n(\cdot,\cdot,\cdot,x)$ is $\Sigma\otimes\mathscr F\otimes\mathscr B(\Delta)$-measurable 
  for all $x\in \mathbb R^\ell_+$, and 
  \[
  C(t,\omega, p)= \bigcap\{C_n(t,\omega,p):n\ge 1\}
  \]
  for all $(t,\omega,p)\in T\times \Omega\times \Delta$. It follows that for all $(t,\omega,p)\in T
  \times \Omega\times \Delta_+$,
  \[
  C^X(t,\omega,p)=\bigcap\{C_n(t,\omega,p):n\ge 1\}\cap X(t, \omega, p).
  \]
  Applying an argument similar to that in Proposition \ref{prop:budcaratheodory}, it can be shown that 
  each $C_n$ is $\Sigma \otimes \mathscr F \otimes \mathscr B(\Delta)$-measurable. Since $X$ is 
  compact-valued, then $C^X$ is $\Sigma\otimes \mathscr F\otimes \mathscr B(\Delta_+)$-measurable.
  \end{proof}
  
  The idea of the next lemma is included in the proof of \cite[Theorem 4.3]{BCY:14}. For the sake of
  self-completeness of this paper, we extracted it here as a separate lemma with a complete proof.
  
  \begin{lemma}\label{lem:CX}
  Assume that an economy $\mathscr E$ satisfies {\rm (A$_1)$}-{\rm (A$_3$)}. Let $\{p_n:n\ge 1\}\subseteq 
  \Delta_+$ converge to some $p\in \Delta_+$. For each $(t,\omega, p) \in T\times \Omega\times \Delta_+$,
  \[
  C^X(t, \omega, p) \subseteq {\rm Li}\ C^X(t, \omega, p_n).
  \]
  \end{lemma}
  
  \begin{proof}
  Let $d\in C^{X} (t,\omega, p)$ be an arbitrarily selected vector. If $d= b(t, \omega, p)$, then 
  $b(t,\omega,p_n)\in C^X (t, \omega, p_n)$ and $\{b(t,\omega,p_n): n\ge 1\}$ converges to $d$. Now, 
  assume $d< b(t, \omega,p)$. Select some $\delta> 0$ and $1\le j \le \ell$ such that
  \[
  d+ (0,\cdots, \underset{j{\rm th}}{\delta},\cdots, 0)\le b(t, \omega, p),
  \]
  and a sequence $\{\delta_i: i\ge 1\}$ in $(0, \delta]$ converging to $0$. For each $i\ge 1$, let
  \[
  d^i= d+ (0,\cdots, \underset{j{\rm th}}{\delta_i},\cdots, 0),
  \]
  and choose a sequence $\{d_n^i: n\ge 1\}$ such that for each $n$, $d_n^i \in X(t,\omega, p_n)$ and 
  $\{d_n^i: n \ge 1\}$ converges to $d^i$. It is claimed that for each $i\ge 1$, $d_n^i\in C^{X}(t, 
  \omega, p_n)$ for sufficiently large $n$. Otherwise, there must exist an $i_0$ and a subsequence
  $\{d_{n_k}^{i_0}: k\ge 1\}$ of $\{d_{n}^{i_0}: n\ge 1\}$ such that $d_{n_k}^{i_0} \notin C^{X}(t, 
  \omega, p_{n_k})$. Let $b_k\in B(t, \omega, p_{n_k})$ and 
  \[
  U(t,\omega, b_k)> U(t,\omega, d_{n_k}^{i_0})
  \] 
  for all $k\ge 1$. Then $\{b_k: k\ge 1\}$ has a subsequence converging to some $b\in B(t, \omega, 
  p)$. By ({A$_3$}), we have 
  \[
  U(t,\omega, b)\ge U(t,\omega,d^{i_0})> U(t,\omega, d),
  \]
  which contradicts with $d\in C^{X}(t, \omega, p)$. It follows from the previous claim that for each 
  $i$, $\{{\rm dist}(d^i, C^X(t,\omega, p_n)): n\ge 1\}$ converges to $0$. Since $\{d^i: i\ge 1\}$ 
  converges to $d$, one concludes that $\{{\rm dist}(d, C^X(t,\omega, p_n)): n \ge 1\}$ converges to 
  $0$. This means that $d \in {\rm Li}\ C^{X}(t, \omega, p_n)$.
  \end{proof}
  
  To conclude this section, we introduce two more correspondences associated with an economy 
  $\mathscr E$ with asymmetric information. For each $\omega \in \Omega$, let ${\mathscr E}(\omega)$ 
  denote the complete information economy, given by
  \[
  {\mathscr E} (\omega) = \left\{(T, \Sigma, \mu); \mathbb R^\ell_+;
  (U(t,\omega,\cdot), a(t, \omega))_{t\in T}\right\}.
  \]
  The \emph{core} of $\mathscr E(\omega)$ is denoted by ${\bf C}(\mathscr E(\omega))$. The set of 
  all Walrasian equilibria and all Walrasian equilibrium allocations of $\mathscr E(\omega)$ are 
  denoted by ${\rm WE}(\mathscr E(\omega))$ and ${\rm WA}(\mathscr E(\omega))$, respectively. Then, 
  ${\bf C}: \omega \mapsto {\bf C}(\mathscr E(\omega))$ and ${\bf WE}: \omega \mapsto {\rm WE}
  (\mathscr E(\omega))$ define two correspondences. 

  \section{The Ex-post Core and Rational Expectations\\ 
  Equilibrium Allocations} \label{sec:ree_vs_epc}
  
  In this section, we discuss the existence of an ex-post core allocation in our model and also
  the relationship between the ex-post core and the set of (Bayesian or maximin) rational
  expectations equilibrium allocations.
  
  \begin{definition}(\cite{Einy-Moreno-Shitovitz:00b})
  Let $f$ be an allocation in an economy $\mathscr E$, let $S \in \Sigma$ be a coalition. We say 
  that $f$ is \emph{ex-post blocked by $S$} if there exist a state of nature $\omega_0 \in \Omega$ 
  and an assignment $g$ such that
  \begin{itemize} 
  \item[(i)] $\displaystyle \int_S g(\cdot, \omega_0)d\mu=\int_S a(\cdot,\omega_0)d\mu$, and
  \item[(ii)] $U(t,\omega_0, g(t,\omega_0))> U(t,\omega_0, f(t,\omega_0))$ $\mu$-a.e. on $S$.
  \end{itemize}
  In addition, an allocation $f$ is called an \emph{ex-post core allocation} if it cannot be 
  ex-post blocked by any coalition. The \emph{ex-post core} of $\mathscr E$, denoted by 
  ${\bf C}({\mathscr E})$, is the set of all the ex-post core allocations of $\mathscr E$.
  \end{definition}
  
  The main result of this section is the following theorem on the ex-post core.
   
  \begin{theorem}\label{thm:mainexpost}
  Suppose that an economy $\mathscr E$ satisfies {\rm (A$_1$)}-{\rm ({A}$_4$)}. Then the 
  ex-post core of $\mathscr E$ is not empty. Moreover, 
  \[
  {\bf C} ({\mathscr E}) = \left\{ f: f \mbox{ is an allocation and } f(\omega, \cdot) \in 
  {\bf C}(\omega) \mbox{ for all } {\omega \in \Omega} \right\}.
  \]
  \end{theorem}
  
  To provide a proof of Theorem \ref{thm:mainexpost}, we need some preparation. First of all, 
  the following result, which is a special case of the Kuratowski-Ryll-Nardzewski measurable 
  selection theorem (refer to \cite[18.13]{Aliprantis-Border:06}), will be needed.
  
  \begin{lemma} \label{lem:Kuratowski}
  Let $F:T\rightrightarrows \mathbb R^\ell$ be a weakly $\Sigma$-measurable 
  correspondence such that $F(t)$ is non-empty and closed for all $t\in T$. Then $F$ 
  admits a $\Sigma$-measurable selection.
  \end{lemma}
  
  Secondly, the following result on the weak measurability of $\bf WE$ is also needed for
  the proof of Theorem \ref{thm:mainexpost}.

  \begin{theorem}\label{thm:W}
  Assume that an economy $\mathscr E$ satisfies {\rm ({A}$_1)$-({A}$_3)$} and 
  {\rm (A$^\prime_4)$}. Then $\bf WE$ is weakly $\mathscr F$-measurable. 
  \end{theorem}
  
  \begin{proof}
  Note that under the given assumptions, ${\bf WE}(\omega)\neq \emptyset$ for all $\omega\in \Omega$. 
  Consider the correspondences $F:(\Omega,\mathscr F, \mathbb P)\rightrightarrows L_1(\mu, \mathbb 
  R^\ell)$, defined by
  \[
  F(\omega)= \left\{f\in L_1(\mu, \mathbb R^\ell): \int_T fd\mu-\int_T a(\cdot,\omega)d\mu= 0\right\}
  \] 
  and $G:(\Omega,\mathscr F, \mathbb P)\rightrightarrows L_1(\mu, \mathbb R^\ell) \times \Delta$, 
  defined by $G(\omega)=F(\omega)\times \Delta$. First of all, we claim that $F$ has a measurable 
  graph, and thus $G$ also has a measurable graph. To see this, define a function $\varphi:(\Omega,
  \mathscr F, \mathbb P) \times L_1(\mu, \mathbb R^\ell)\to \mathbb R^\ell$ by
  \[
  \varphi(\omega, f)=\int_T fd\mu- \int_T a(\cdot,\omega)d\mu.
  \] 
  Note that for every $f\in L_1(\mu, \mathbb R^\ell)$, $\varphi(\cdot, f)$ is $\mathscr F$-measurable 
  and for every $\omega\in \Omega$, $\varphi(\omega, \cdot)$ is norm-continuous. Thus, $\varphi$
  is ${\mathscr F}\otimes {\mathscr B}(L_1(\mu, \mathbb R^\ell))$-measurable. The conclusion follows 
  from the fact ${\rm Gr}_F= \varphi^{-1}(0)$.
  
  \medskip
  Let $\mathbb Q^n\cap \Delta_+ = R$, where ${\mathbb Q}^n$ is the set of vectors in $\mathbb R^\ell$
  with rational components. Note that $R$ is countable and dense in $\Delta$. For each $p \in R$, 
  define a correspondence $H_{p}:(\Omega,\mathscr F, \mathbb P)\rightrightarrows L_1(\mu, \mathbb 
  R^\ell)$ by $H_{p}(\omega)= {\mathscr S}_{C^X(\cdot, \omega, p)}$, where ${\mathscr S}_{C^X
  (\cdot,\omega,p)}$ is the set of integrable selections of $C^X(\cdot,\omega,p)$. Fix a $p\in 
  R$ and define a function $\zeta:L_1(\mu, \mathbb R^\ell)\times (\Omega,\mathscr F,
  \mathbb P)\to \mathbb R_+$ by $\zeta(g,\omega)= {\rm dist}(g, H_{p}(\omega))$. Furthermore, 
  for each $g\in L_1(\mu,\mathbb R^\ell)$, let the function $\xi^g: (T,\Sigma,\mu) \times 
  (\Omega, \mathscr F, \mathbb P) \to \mathbb R_+$ be defined by
  \[
  \xi^g (t, \omega)= {\rm dist} \left(g(t), C^X(t, \omega, p)\right).
  \] 
  {\bf Claim 1.} \emph{For each simple function $g \in L_1(\mu, \mathbb R^\ell)$, $\displaystyle 
  \zeta(g,\omega) = \int_T\xi^g(\cdot, \omega)d\mu$ holds for all $\omega \in \Omega$.}
  
  \begin{proof}[Proof of Claim 1]
  Let $g \in L_1(\mu,\mathbb R^\ell)$ be a given simple measurable function. 
  As $g$ is a step-function with finitely many values, it follows from Lemma 
  \ref{lem:measurable} and Proposition \ref{prop:fmeasure} that $\xi^g$ is 
  $\Sigma\otimes {\mathscr F}$-measurable. In addition, since 
  \[
  \xi^g (t, \omega)\le \|g(t)-b(t,\omega,p)\|
  \]
  for all $(t,\omega)\in T\times \Omega$, $\xi^g(\cdot, \omega)$ is also integrably 
  bounded and thus $\xi^g(\cdot, \omega) \in L_1(\mu, \mathbb R^\ell)$ for all 
  $\omega \in \Omega$. It is easy to check that $\int_T \xi^g(\cdot, \omega)d\mu 
  \le \zeta(g, \omega)$ for all $\omega \in \Omega$. Suppose $\int_T \xi^g(\cdot, 
  \omega_0) d\mu <\zeta(g,\omega_0)$ holds for some $\omega_0\in \Omega$. Then, 
  there is an $\varepsilon> 0$ such that 
  \[
  \int_T \xi^g(\cdot,\omega_0)d\mu+ \varepsilon\mu(T)<\zeta(g,\omega_0).
  \]
  Next, we define $A:T\rightrightarrows \mathbb R^\ell$ and 
  $\alpha:T \times \mathbb R^\ell\to \mathbb R$ by
  \[
  A(t)= \left\{y\in C^X(t,\omega_0, p):\|g(t)- y\| \le \xi^g(t,\omega_0)+ 
  \varepsilon\right\}
  \]
  and
  \[
  \alpha(t, y)= \|g(t)- y\|-\xi^g(t,\omega_0).
  \]
  As done in the above, it can be shown that $\alpha$ is $\Sigma \otimes 
  {\mathscr F}$-measurable and thus
  \[
  {\rm Gr}_{A}= \left\{(t,y)\in T\times \mathbb R^\ell: \alpha(t,y)\le 
  \varepsilon\right\}\cap {\rm Gr}_{C^X(\cdot,\omega_0,p)}
  \]
  is measurable. By Lemma \ref{lem:Kuratowski}, $A$ has a measurable selection 
  $h: T\to \mathbb R^\ell$ satisfying
  \[
  \|g- h\|_{L_1}\le \int_T \xi^g(\cdot, \omega_0)d\mu+ \varepsilon\mu(T).
  \]
  As $h \in H_p(\omega_0)$, we have
  \[
  \zeta(g, \omega_0) \le \int_T \xi^g(\cdot, \omega_0)d\mu+ \varepsilon\mu(T),
  \]
  which is a contradiction.
  \end{proof}
  
  \medskip
  \noindent
  {\bf Claim 2.} \emph{For each function $g \in L_1(\mu, \mathbb R^\ell)$, $\displaystyle 
  \zeta(g,\omega) = \int_T \xi^g(\cdot, \omega)d\mu$ holds for all $\omega \in \Omega$, and 
  thus $H_{p}$ is weakly $\mathscr F$-measurable.}
  
  \begin{proof}[Proof of Claim 2]
  Let $\{g_n:n\ge 1\}$ be a sequence of simple measurable functions converging to $g$ in $L_1
  (\mu,\mathbb R^\ell)$. By \cite[Theorem 13.6]{Aliprantis-Border:06}, there is a subsequence 
  $\{g_{n_k}:k\ge 1\}$ of $\{g_n:n\ge 1\}$ and a function $h\in L_1(\mu,\mathbb R^\ell)$ such 
  that $|g_{n_k}|\le h$ for all $k\ge 1$ and $\{g_{n_k}:k\ge 1\}$ converges pointwise to $g$. 
  Thus, $\{\xi^{g_{n_k}}(t,\omega): k\ge 1\}$ converges to $\xi^g(t,\omega)$ for all $(t,\omega)
  \in T\times \Omega$. As 
  \[
  \xi^{g_{n_k}}(t,\omega)\le \|h(t)+b(t,\omega,p)\|,
  \]
  we have $\{\xi^{g_{n_k}}(\cdot,\omega):k\ge 1\}$ is dominated by the integrable function 
  $h+b(\cdot,\omega,p)$. Hence, by the Lebesgue dominated convergence theorem, we have 
  \[
  \lim_{k\to \infty}\int_T \xi^{g_{n_k}}(\cdot,\omega) d\mu=\int_T \xi^{g}(\cdot,\omega)d\mu 
  \]
  for all $\omega\in \Omega$. On the other hand, $\{\zeta(g_{n_k},\omega):k\ge 1\}$ converges to 
  $\zeta(g,\omega)$ for every $\omega \in \Omega$. Thus, we have
  \[
  \zeta(g,\omega) = \int_T \xi^{g}(\cdot,\omega)d\mu
  \]
  for all $\omega\in \Omega$. Moreover, since each $\zeta(g_{n_k},\cdot)$ is $\mathscr F$-measurable, 
  we have $\zeta(g,\cdot)$ is $\mathscr F$-measurable. Thus, $H_{p}$ is weakly $\mathscr F$-measurable.  
  \end{proof}
  
  Define a correspondence $H:(\Omega,\mathscr F, \mathbb P) \rightrightarrows L_1(\mu, \mathbb R^\ell)
  \times \Delta$ by letting
  \[
  H(\omega)=\overline{\bigcup\{H_{p}(\omega)\times \{p\}: p\in R\}}
  \]
  for all $\omega\in \Omega$, where the closure operation is taken in the product topology on $L_1
  (\mu,\mathbb R^\ell) \times \Delta$ induced by the norm of $L_1(\mu, \mathbb R^\ell)$ and the 
  norm of $\mathbb R^\ell$. 
  	
  \bigskip
  \noindent
  {\bf Claim 3.} \emph{${\bf WE}(\omega)= H(\omega)\cap G(\omega)$ for all $\omega\in \Omega$.}
  	
  \begin{proof}[Proof of Claim 3]
  Fix some $\omega\in \Omega$. Let $(f, p)\in {\bf WE}(\omega)$. Clearly, $(f, p)\in  G(\omega)$. By 
  ({A}$_3)$, we have $p\in \Delta_+$. It follows that $f(t)\in C^X(t,\omega,p)$ $\mu$-a.e. on $T$. 
  Now, suppose that $\{p_n: n\ge 1\}\subseteq R$ is a sequence converging to $p$. By Claim 2, 
  \[
  {\rm dist}(f, H_{p_n}(\omega))= \int_T \eta_n d\mu,
  \] 
  where $\eta_n:T\to \mathbb R_+$ is defined by $\eta_n(t)= {\rm dist}(f(t), C^X(t,\omega, p_n))$. By 
  Lemma \ref{lem:CX}, 
  \[
  C^X(t,\omega,p) \subseteq {\rm Li}\ C^{X}(t, \omega, p_n)
  \] 
  for all $t\in T$. Hence,  for each $t\in T$ and each $n\ge 1$, we can choose some $f_n(t)\in 
  C^X(t,\omega,p_n)$ such that $f_n(t) \to f(t)$, $\mu$-a.e on $T$. It follows that $\eta_n (t) \to 0$, 
  $\mu$-a.e on $T$. Define
  \[
  \beta= \inf\left(\{\delta(p_n): n\ge 1\}\cup \{\delta(p)\}\right).
  \]
  Then $\beta> 0$ and for each $t\in T$, let 
  \[
  d(t)= \frac{1}{\beta}\sum_{h= 1}^\ell a^h (t,\omega).
  \]
  Note that 
  \[
  \eta_n(t) \le \|f_n(t) - f(t) \| \le 2\|d(t){\bf 1}\|,
  \] 
  $\mu$-a.e on $T$ for all $n\ge 1$. By the Lebesgue dominated convergence theorem again, we have 
  \[
  {\rm dist}(f, H_{p_n}(\omega)) = \int_T \eta_n d\mu \to 0.
  \]
  It follows that $(f, p)\in H(\omega)$.
  		
  \medskip
  Let $(f, p)\in H(\omega)\cap G(\omega)$ for an arbitrarily fixed $\omega \in \Omega$. Analogous to 
  the proof of Claim 2, we can find a sequence $\{r_n: n\ge 1\} \subseteq R$ and $f_n\in H_{r_n}
  (\omega)$ such that $(f_n, r_n)\to (f, p)$ in $L_1(\mu,\mathbb R^\ell)\times \Delta$, as $n\to 
  \infty$ and $\{f_n:n\ge 1\}$ pointwise converges to $f$. Since $p\in \Delta$, by (A$_1$), we must 
  have $\langle p, \int_T a(\cdot,\omega)d\mu\rangle> 0$. Put,
  \[
  S=\left\{t\in T: \langle p, a(t,\omega)\rangle> 0\right\}.
  \]
  Definitely, $S\in \Sigma$ and $\mu(S)> 0$. Define 
  \[
  A_n=\{t\in S: f_n(t)\notin C^X(t,\omega,r_n)\} \mbox{ and } A=\bigcup\{A_n:n\ge 1\}.
  \]
  Since $\mu(A_n)=0$ for all $n\ge 1$, one must have $\mu(A)=0$. Choose a $t\in S\setminus A$. If 
  $f(t)\notin C(t,\omega,p)$ for some $t\in S\setminus A$, by ({A}$_3)$, there must exist an element 
  $y\in \mathbb R^\ell_{+}$ such that $\langle p, y\rangle< \langle p, a(t,\omega)\rangle$ and 
  $U(t,\omega,y)> U(t,\omega,f(t))$. As a result, $\langle r_n, y\rangle< \langle r_n, a(t,\omega)
  \rangle$ and $U(t,\omega,y)> U(t,\omega,f_n(t))$ for sufficiently large $n$, which is a contradiction. 
  Thus, $f(t)\in C(t,\omega,p)$ for all $t\in S\setminus A$. Since $U(t, \omega, \cdot)$ is strictly 
  increasing, $\langle p, f(t)\rangle\ge \langle p, a(t,\omega)\rangle$ $\mu$-a.e. on $S$. Moreover, 
  $\langle p, f(t)\rangle\ge 0=\langle p, a(t,\omega)\rangle$ for all $t\in T\setminus S$. Hence,
  $\langle p, f(t)\rangle\ge \langle p, a(t,\omega)\rangle$ $\mu$-a.e. on $T$, which together with
  the feasibility of $f$ implies that $f(t)\in B(t,\omega,p)$ $\mu$-a.e. on $T$. If $\mu(T
  \setminus S)=0$, then $(f, p)\in {\bf WE}(\omega)$. Otherwise, we first claim that $p\in \Delta_+$. 
  If not, there is some $z> 0$ such that $\langle p, z\rangle=0$. Consequently, $f(t)+z\in B(t,\omega, 
  p)$ and 
  \[
  U(t,\omega, f(t)+z)> U(t,\omega,f(t))
  \] 
  for all $t\in S\setminus A$, which is a contradiction. So, $B(t,\omega,p)= \{0\}$ and $f(t)=0$ 
  for $\mu$-a.e. on $T\setminus S$. Thus, $(f, p) \in {\bf WE}(\omega)$. 
  \end{proof}
  	
  By Claim 2, ${\rm Gr}_{\bf WE}= {\rm Gr}_{H}\cap {\rm Gr}_{G}$. Since both ${\rm Gr}_{H}$ and ${\rm 
  Gr}_{G}$ are measurable, ${\rm Gr}_{\bf WE}$ is measurable. Hence, $\bf WE$ is weakly $\mathscr 
  F$-measurable.       
  \end{proof}
  
  Now, we are ready to provide a proof of  Theorem \ref{thm:mainexpost}, as promised previously.

  \begin{proof}[Proof of Theorem \ref{thm:mainexpost}]
  First of all, for the sake of convenience, we put
  \[
  X= \left\{ f: f \mbox{ is an allocation and } f(\cdot, \omega) \in {\bf C}(\mathscr E(\omega)) 
  \mbox{ for all } {\omega \in \Omega} \right\}.
  \]
  It is easy to see that ${\bf WE}$ is non-empty closed-valued. Then, following Theorem \ref{thm:W}, 
  ${\bf WE}$ is weakly $\Sigma$-measurable. By Lemma \ref{lem:Kuratowski}, $\bf WE$ has a measurable 
  selection $\omega \mapsto (f(\cdot,\omega), \pi(\omega))$. Note that $\pi(\omega)\in \Delta_+$ for 
  all $\omega\in \Omega$. Under assumption ({A}$_4)$, $B(t,\omega,\pi(\omega))\cap C^X(t,\omega,
  \pi(\omega))$ is singleton for all $(t,\omega) \in T\times \Omega$. Let $g:T\times \Omega \to 
  \mathbb R^\ell_+$ be the function defined by 
  \[
  g(t,\omega)= B(t,\omega,\pi(\omega))\cap C^X(t,\omega,\pi(\omega))
  \] 
  for all $(t,\omega)\in T\times \Omega$. By Proposition \ref{prop:budcaratheodory} and Proposition 
  \ref{prop:fmeasure}, $g$ is $\Sigma\otimes \mathscr F$-measurable. Hence, $g(t,\cdot)$ is $\mathscr 
  F$-measurable for all $t\in T$. For all $\omega\in \Omega$, $g(\cdot,\omega)=f(\cdot,\omega)$ 
  $\mu$-a.e., which implies that $(g(\cdot,\omega), \pi(\omega))\in {\bf WE}(\omega)$ for 
  all $\omega\in \Omega$. It follows that $g(\cdot,\omega)\in {\bf C}(\mathscr E(\omega))$ for all 
  $\omega\in \Omega$. Hence, $X \neq \emptyset$. Then, it is straightforward to see that 
  $X \subseteq {\bf C}({\mathscr E})$.  
  
  \medskip
  To show that ${\bf C}(\mathscr E) \subseteq X$, we suppose that there exists an $f\in 
  {\bf C}(\mathscr E) \setminus X$. Then there exists a state $\omega_0 \in \Omega$ such that 
  $f(\cdot, \omega_0) \not \in {\bf C}({\mathscr E}(\omega_0))$. This means that $f$ is blocked
  in ${\mathscr E}(\omega_0)$ by some coalition $S$. Therefore, there exists an assignment 
  $g: T \to \mathbb R^\ell_+$ in ${\mathscr E}(\omega_0)$ such that
  \begin{itemize} 
  \item[(i)] $\displaystyle \int_S g d\mu=\int_S a(\cdot,\omega_0)d\mu$, and
  \item[(ii)] $U(t,\omega_0, g(t))> U(t,\omega_0, f(t,\omega_0))$ $\mu$-a.e. on $S$.
  \end{itemize}
  Define
  \[
  \Omega_0 = \left\{ \omega\in \Omega: \int_S g d\mu = \int_S a(\cdot,
  \omega) d\mu \right\}.
  \]
  Obviously, $\omega_0 \in \Omega_0$ and $\Omega_0 \in \mathscr F$. Define a 
  function $h:T\times \Omega \to \mathbb R^\ell_+$ by
  \[
  h(t, \omega)= \left\{
  \begin{array}{ll}
  g(t), & \mbox{if $(t, \omega)\in S \times \Omega_0$;}\\[0.3em]
  a(t,\omega), & \mbox{otherwise}.
  \end{array}
  \right.
  \]
  Then, it can be readily checked that $h$ is an assignment in $\mathscr E$ such that
  \begin{itemize} 
  \item[(iii)] $\displaystyle \int_S h(\cdot, \omega_0)d\mu=\int_S a(\cdot,\omega_0)d\mu$, and
  \item[(iv)] $U(t,\omega_0, h(t,\omega_0))> U(t,\omega_0, f(t,\omega_0))$ $\mu$-a.e. on $S$.
  \end{itemize}
  This means that $f$ is ex-post blocked by $S$ (via an assignment $h$ at the state $\omega_0$), 
  which contradicts with the fact of $f\in {\bf C}(\mathscr E)$. 
  \end{proof}
  
  Next, we discuss the consequences of Theorems \ref{thm:mainexpost} and \ref{thm:W}. We 
  need to introduce two competitive equilibrium concepts in the economy model $\mathscr E$ discussed  
  in Subsection \ref{subsec: model}: maximin rational expectations equilibrium and Bayesian rational 
  expectations equilibrium. Given an agent $t \in T$, a state of nature $\omega\in \Omega$ and a 
  price system $\pi: \Omega \to \Delta$, let $B^{REE}(t,\omega,\pi)$ be defined by
  \[
  B^{REE}(t, \omega, \pi)= \left\{x\in (\mathbb R^\ell_+)^\Omega:
  x(\omega^\prime)\in B(t, \omega^\prime, \pi(\omega^\prime))
  \mbox{ for all } \omega^\prime \in \mathscr G_t(\omega)\right\}.
  \]
  The \emph{maximin utility} of each agent $t\in T$ with respect to $\mathscr G_t$ at $x: \Omega
  \to \mathbb R^\ell_+$ in state $\omega \in \Omega$, denoted by $\b{\it U}^{REE}(t, \omega, x)$, is 
  defined by
  \[
  \b{\it U}^{REE}(t,\omega, x)= \inf\left\{ U(t, \omega', x(\omega')):
  \omega'\in \mathscr G_t (\omega)\right\}.
  \]
  
  \begin{definition}[\cite{de Castro-Pesce-Yannelis:11}] \label{def:maximinree}
  Given an allocation $f$ and a price system $\pi$ in an economy $\mathscr E$, the pair $(f, \pi)$ 
  is called a \emph{maximin rational expectations equilibrium} (abbreviated as maximin REE) of 
  $\mathscr E$ if $f(t, \omega) \in B(t, \omega, \pi(\omega))$ and $f(t, \cdot)$ maximizes 
  $\b{\it U}^{REE}(t, \omega, \cdot)$ on $B^{REE}(t, \omega, \pi)$ for all $(t, \omega) \in T 
  \times \Omega$. In this case, $f$ is called a \emph{maximin rational expectations allocation}, 
  and the set of such allocations is denoted by $MREE(\mathscr E)$.
  \end{definition}
  
  Define $L_t^{REE}$ by
  \[
  L_t^{REE} = \{x \in (\mathbb R_+^\ell)^\Omega: x \mbox{ is }{\mathscr G_t}\mbox{-measurable}\}.
  \]
  For a given $x\in L_t^{REE}$, recall that the \emph{Bayesian expected utility} of agent $t$ 
  with respect to $\mathscr G_t$ at $x$ is given by ${\mathbb E}_t\left[ U(t, \cdot,x)|{\mathscr 
  G}_t \right]$. 
  
  \begin{definition}[\cite{Allen:81, Radner:79}] \label{def:bayesianree}
  Given an allocation $f$ and a price system $\pi$ in an economy $\mathscr E$, the pair $(f, \pi)$ 
  is called a \emph{Bayesian rational expectations equilibrium} (abbreviated as Bayesian REE) of 
  $\mathscr E$ if
  \begin{enumerate}
  \item[(i)] for each $t\in T$, $f(t,\cdot)$ is $\mathscr G_t$-measurable;
  \item[(ii)] for all $(t,\omega) \in T\times \Omega$, $f(t, \omega)\in B(t,\omega,\pi(\omega))$;
  \item[(iii)] for all $(t,\omega) \in T\times \Omega$,
  \[
  {\mathbb E}_t\left[ U(t, \cdot, f(t,\cdot)) | {\mathscr G}_t \right](\omega) = 
  \max_{x \in B^{REE}(t,\omega, \pi) \cap L_t^{REE}}{\mathbb E}_t\left[ U(t,\cdot, x) | 
  {\mathscr G}_t \right](\omega),
  \]
  \end{enumerate}
  In this case, $f$ is called a \emph{rational expectations allocation}, and the set of such 
  allocations is denoted by $REE(\mathscr E)$.
  \end{definition}
  
  As a corollary of Theorem \ref{thm:W}, we can retrieve the result on the existence of a
  maximin REE or Bayesian REE obtained in \cite{Bhowmik-Cao:16}.
  
  \begin{proposition}
  If an economy $\mathscr E$ satisfies assumptions {\rm ({A}$_1)$-({A}$_4)$}, then we have
  $MREE(\mathscr E) = REE(\mathscr E) \neq \emptyset$.
  \end{proposition}
  
  \begin{proof}
  Note that under the assumption ${\bf WE}(\omega)\neq \emptyset$ for all $\omega\in \Omega$. From 
  the proof of Theorem \ref{thm:W}, we can see that ${\bf WE}$ is closed-valued and weakly 
  $\Sigma$-measurable, thus it has a $\mathscr F$-measurable selection $\omega\mapsto (f(\omega),
  \pi(\omega))$. Then, it is easy to verify that the pair $(g, \pi)$ defined in the proof of
  Theorem \ref{thm:mainexpost} is a maximin rational expectations equilibrium of $\mathscr E$. 
  \end{proof}
  
  As a consequence of Theorem \ref{thm:mainexpost} and \cite[Corollary 1]{Bhowmik-Cao:16}, 
  we deduce the following version of the first fundamental theorem of social welfare for our model.
  
  \begin{corollary} \label{coro:welfare}
  Suppose that an economy $\mathscr E$ satisfies {\rm (A$_1$)}-{\rm ({A}$_4$)}. Then, we have
  $MREE(\mathscr E) = REE(\mathscr E) \subseteq {\bf C}({\mathscr E})$.  
  \end{corollary}
  
  Next, we provide an example of a continuum economy $\mathscr E$ with asymmetric information and  
  infinitely many states of nature in which the ex-post core ${\bf C}(\mathscr E)$ can 
  strictly contain $REE(\mathscr E)$.
   	
  \begin{example} \label{exam:sinclusion}
  Consider an economy $\mathscr E$ defined by
  \[
  \mathscr E= \left\{(\Omega, \mathscr F, \mathbb P); (T,\Sigma,\mu); {\mathbb R}^\ell_+; ({\mathscr 
  F}_t, U(t,\cdot,\cdot), a(t, \cdot), \mathbb P_t)_{t\in T}
  \right\},
  \]
  where $T=\Omega=[0, 1]$, $\Sigma$ and $\mathscr F$ are the Borel $\sigma$-algebra on $[0, 1]$, $\mu$
  and $\mathbb P$ are the Lebesgue probability measure. The commodity space is $\mathbb R^2$. Let 
  $\mathscr F_t$ and $\mathbb P_t$ be arbitrary information partition and the prior belief of agent 
  $t\in T$. The utility and the initial endowment of each agent are given by $U(t, \omega, x) =
  \sqrt x_1+ \sqrt x_2$ and
  \[
  a(t,\omega) = \left\{
  \begin{array}{ll}
  (1, 2), & \mbox{if $(t,\omega)\in \left[0, \frac{1}{2}\right]\times \Omega$}; \\[0.5em]
  (3, 2), & \mbox{if $(t,\omega)\in \left(\frac{1}{2}, 1\right]\times \Omega$},
  \end{array}
  \right.
  \]
  respectively. Then $\mathscr E(\omega)= \mathscr E(\omega')$ for all $\omega, \omega'\in \Omega$. 
  Furthermore, it can be easily checked that {\rm (A$_1$)}-{\rm ({A}$_4$)} are satisfied.
	
  \medskip
  For every $p= (p_1, p_2)\in \Delta_+$, the demand of agent $t$ in each state is given by
  \[
  D(t,\omega, p) = \left\{
  \begin{array}{ll}
  \left(\frac{p_2(1+p_2)}{p_1}, \frac{p_1(1+p_2)}{p_2}\right), & \mbox{if 
  $(t,\omega)\in \left[0, \frac{1}{2}\right]\times \Omega$};\\[0.5em]
  \left(\frac{p_2(2+p_1)}{p_1}, \frac{p_1(2+p_1)}{p_2}\right),& \mbox{if 
  $(t,\omega)\in \left(\frac{1}{2}, 1\right]\times \Omega$}.
  \end{array}
  \right.
  \]
  By the market-clearing condition, we can show that the equilibrium price is $p_0=(\frac{1}{2}, 
  \frac{1}{2})$. Thus, $D(t,\omega,p_0) =\left(\frac{3}{2}, \frac{3}{2}\right)$ for all $(t,\omega) 
  \in \left[0, \frac{1}{2}\right] \times \Omega$ and $D(t,\omega,p_0)=\left(\frac{5}{2}, \frac{5}{2}
  \right)$ for all $(t, \omega) \in \left(\frac{1}{2},1\right] \times \Omega$. For each $\omega \in 
  \Omega$, let $h: T \to \mathbb R^2_+$ be an allocation in $\mathscr E(\omega)$ defined by
  \[
  h(t)= \left\{
  \begin{array}{ll}
  \left(\frac{3}{2}, \frac{3}{2}\right), & \mbox{ if $t\in \left[0,
  \frac{1}{2}\right]$}; \\[0.5em]
  \left(\frac{5}{2}, \frac{5}{2}\right), & \mbox{ if $t\in
  \left(\frac{1}{2}, 1\right]$}.
  \end{array}
  \right.
  \]
  Take a subset $A$ of $\Omega$ with $A\notin \mathscr F$, and consider $f:T\times \Omega \to \mathbb 
  R_+^2$, defined by 
  \[
  f(t,\omega) = \left\{
  \begin{array}{ll}
  (1, 1), & \mbox{if $t\in A$ and $t=\omega$}; \\[0.5em]
  h(t),   & \mbox{otherwise}.
  \end{array}
  \right.
  \] 
  (i) It is clear that $f$ is feasible.\\
  (ii) For each $t\in T$, $f(t,\cdot)$ is $\mathscr F$-measurable. To see this, we first choose 
  $t\in A$. In this case, $f(t,\omega)=h(t)$ for all $\omega\in \Omega$ with $\omega\neq t$; and 
  $f(t,\omega)= (1, 1)$ if $\omega=t$. Now, take $t\in T\setminus A$, then $f(t,\omega)=h(t)$ for 
  all $\omega\in \Omega$. Thus, in both cases, $f(t,\cdot)$ is $\mathscr F$-measurable.\\
  (iii) It is clear that for each $\omega \in \Omega$, $f(\cdot, \omega)$ is $\mu$-integrable. 
	
  \medskip
  Since $(f(\cdot,\omega), p_0) \in {\bf WE}(\omega)$ for each $\omega \in \Omega$, we 
  conclude that $f \in {\bf C}({\mathscr E}(\omega))$ and thus $f \in {\bf C}({\mathscr E})$.
  Now, consider two mappings $g_1: t\mapsto (t,t)$ and $g_2:(t,t)\mapsto t$ defined by 
  $g_1(t)=(t, t)$ and $g_2(t, t)=t$, respectively. It can be readily checked that
  \[
  (g_2\circ f\circ g_1)(t) = \left\{
  \begin{array}{ll}
  1, \mbox{ if $t\in A$}; \\[0.3em]
  e(t), \mbox{ otherwise,}
  \end{array}
  \right.
  \]
  where $e(t)= \frac{3}{2}$ if $t\in \left[0, \frac{1}{2}\right]$; and $e(t)= \frac{5}{2}$ if $t\in 
  \left(\frac{1}{2}, 1\right]$. Since $A \not \in \mathscr F$, then $g_2\circ f\circ g_1$ is not 
  $\mathscr F$-measurable. It follows that $f$ is not $\Sigma\otimes \mathscr F$-measurable. 
  By Corollary 1 in \cite{Bhowmik-Cao:16}, ${\rm REE}(\mathscr E) \ne \emptyset$ and $f\notin 
  {\rm REE}(\mathscr E)$.  
  \end{example}

  \section{The Ex-post Core and the Fine Core} \label{sec:FineCore}
  
  In this section, we study the relationship between the ex-post core and the fine core in a mixed
  economy with asymmetric information. We show that under appropriate assumptions, the fine core
  is contained in the ex-post core (see Theorem \ref{thm:Ex-postCore}). This extends a result of Einy 
  et al. in \cite{Einy-Moreno-Shitovitz:00}. To achieve this goal, we use a standard approach, which 
  embeds the original mixed economy $\mathscr E$ into the auxiliary atomless economy ${\mathscr E}^*$ 
  obtained by splitting each large agent into a continuum of small agents of the same type.  
  
  \subsection{Interpretation via associated continuum economies}\label{subsec:continuum}
  We define an atomless economy ${\mathscr E}^*$ associated with $\mathscr E$. Let $(T_1^\ast,
  \Sigma^\ast_{T_1}, \mu^\ast_{T_1})$ be an atomless, complete and positive measure space such that 
  $T_0\cap T_1^\ast= \emptyset$, where each agent $A_n$ one-to-one corresponds to a measurable subset
  $A^\ast_n$ of $T_1^\ast$ with $\mu^\ast(A_n^\ast)= \mu(A_n)$ and $T_1^\ast= \bigcup \{ A_n^*: n 
  \ge 1 \}$. One can think that $T_1^*$ is constructed as follows: Partition the interval $\left[
  \mu(T_0), \mu(T) \right]$, which is identified with $T_1^*$, as the disjoint union of the intervals
  $A_n^*$ given by $A_1^* = \left[\mu(T_0), \mu(T_0) + \mu(A_1)\right), \cdots$, and
  \[
  A_n^* = \left[\mu(T_0) + 
  \mu\left(\bigcup_{i=1}^{n-1}A_i\right),\  \mu(T_0) + \mu\left(\bigcup_{i=1}^nA_i\right)\right), 
  \cdots.
  \]
  Define $T^\ast= T_0\cup T_1^\ast$ with the $\sigma$-algebra
  \[
  \Sigma^\ast= \Sigma_{T_0} \oplus \Sigma^\ast_{T_1}= \left\{A\cup B:A\cap 
  B=\emptyset, \ A\in \Sigma_{T_0}, B\in
  \Sigma^\ast_{T_1} \right\}
  \]
  and the measure $\mu^\ast:\Sigma^\ast\to \mathbb R_+$ such that for each $C \in \Sigma^*$,
  \[
  \mu^\ast(C)= \mu_{T_0}(C\cap T_0)+ \mu^\ast_{T_1}(C\cap
  T_1^*).
  \]
  Following \cite{Bhowmik-Cao:paper3, Pesce:10}, the space of agents of $\mathscr E^\ast$ is $(T^\ast, 
  \Sigma^\ast, \mu^\ast)$. In addition, in ${\mathscr E}^*$, the space of states of nature and the 
  consumption set for each agent $t\in T^\ast$ at each state $\omega\in \Omega$ are still $(\Omega,
  \mathscr F,\mathbb P)$ and $\mathbb R^\ell_+$, respectively. Finally, the characteristics 
  $({\mathscr F}^*_t, U^*(t,\cdot, \cdot), a^*(t,\cdot), {\mathbb P}^*_t)$ of each agent $t \in 
  T^*$ in ${\mathscr E}^*$ are defined as follows:
  \[
  {\mathscr F}^*_t = \left\{
  \begin{array}{ll}
  {\mathscr F}_t, & \mbox{if $t\in T_0$;}\\[0.3em]
  {\mathscr F}_{A_n}, & \mbox{if $t\in A_n^*$,}
  \end{array}
  \right.
  \]
  \[
  U^*(t, \omega, \cdot) = \left\{
  \begin{array}{ll}
  U(t,\omega,\cdot), & \mbox{if $(t, \omega) \in T_0 \times \Omega$;}\\[0.3em]
  U(A_n, \omega, \cdot), & \mbox{if $(t,\omega)\in A_n^*
  	\times \Omega$,}
  \end{array}
  \right.
  \]
  \[
  a^* (t, \omega) = \left\{
  \begin{array}{ll}
  a(t,\omega), & \mbox{if $(t, \omega) \in T_0 \times
  	\Omega$;}\\[0.3em]
  a(A_n, \omega), & \mbox{if $(t,\omega)\in A_n^* \times
  	\Omega$,}
  \end{array}
  \right.
  \]
  and
  \[
  {\mathbb P}^*_t = \left\{
  \begin{array}{ll}
  {\mathbb P}_t, & \mbox{if $t\in T_0$;}\\[0.3em]
  {\mathbb P}_{A_n}, & \mbox{if $t\in A_n^*$.}
  \end{array}
  \right.
  \]
  For each $\omega \in \Omega$, we can define an atomless and deterministic economy ${\mathscr 
  E}^*(\omega)$ associated with ${\mathscr E}(\omega)$ as
  \[
  {\mathscr E}^*(\omega) = \left\{(T^*, \Sigma^*, \mu^*); \mathbb R^\ell_+; (U^*(t,\omega,\cdot), 
  a^*(t, \omega))_{t\in T^*}\right\}.
  \]
  Similar to that of $\mathscr E$, we call a member $S$ of $\Sigma^*$ with $\mu^*(S)>0$ a coalition
  of $\mathscr E^*$. Given a coalition $S \in \Sigma^*$ of $\mathscr E^*$, let $\Sigma^*_S= \{A\in 
  \Sigma^*: A\subseteq S\}$. 
  
  \medskip
  The following lemma is a particular case of \cite[Lemma 3.6]{Bhowmik-Cao:paper3}.
  
  \begin{lemma}[\cite{Bhowmik-Cao:paper3}] \label{lemma:lyapunov}
  Given $\omega \in \Omega$, if $f\in L_1 \left(\mu^*, \mathbb R^\ell\right)$ and $S, R$ are two 
  coalitions of ${\mathscr E}^*(\omega)$ such that $\mu^\ast(S\cap R)> 0$, then
  \[
  H= {\rm cl} \left\{\left(\mu^*(B), \int_B f d\mu^* \right): B\in \Sigma^*_S \right\}
  \]
  is a convex subset of ${\mathbb R} \times \mathbb R^\ell$. Moreover, for any $0< \delta< 1$, 
  there is a sequence $\{C_n: n\ge 1\} \subseteq \Sigma^*_S$ of coalitions in $\mathscr E^*$ such 
  that $\mu^*(C_n\cap R)=\delta \mu^*(S\cap R)$ for all $n \ge 1$ and
  \[
  \lim_{n\to \infty} \int_{C_n} f(\cdot,\omega) d\mu^*=\delta \int_S f(\cdot, \omega) d\mu^*.
  \]
  \end{lemma}
  
  \begin{lemma} \label{lem:privateblocked}
  Assume that an economy $\mathscr E$ satisfies $({\rm A}'_1)$, $({\rm A}_2)$ and $({\rm A}_3)$. Let  
  $\omega \in \Omega$ be a state of nature. If an allocation $f^\ast$ in $\mathscr E^\ast(\omega)$ 
  is blocked by a coalition $S \subseteq T^*$, then for any $0<\varepsilon\le \mu^\ast(S\cap T_1^\ast)$, 
  there exist a coalition $R^*$ such that
  \[
  R^\ast\subseteq \bigcup\left\{S^i: i\in \mathfrak P(S) \right\},
  \]
  $\mu^\ast(R^\ast\cap T_1^\ast)= \varepsilon$ and $\mu(R^\ast\cap S^i)> 0$ for all $i\in \mathfrak 
  P(S)$, and an allocation $g^\ast$ in $\mathscr E^\ast(\omega)$ such that $f^\ast$ is blocked by 
  $R^\ast$ via $g^\ast$ in $\mathscr E^\ast(\omega)$.
  \end{lemma}
  
  \begin{proof}
  If $\varepsilon= \mu^\ast(S\cap T_1^\ast)$, there is nothing to prove. So, let $0< \varepsilon< 
  \mu^\ast (S\cap T_1^\ast)$. By the techniques in \cite[Lemma 3.5 ]{Bhowmik-Cao:paper3}, we can 
  find a function $h^\ast:T^\ast\to \mathbb R^\ell_+$ such that 
  \[
  U^\ast(t, \omega, h^\ast(t))> U^\ast (t,\omega,f^\ast(t)),\ \mbox{$\mu$-a.e on $S$}
  \] 
  and
  \[
  \int_{S}(a^\ast(\cdot,\omega)-h^\ast(\cdot))d\mu^\ast\gg 0.
  \]
  Let $\delta= \frac{\varepsilon}{\mu^\ast(S\cap T_1^\ast)}$. For each $i\in \mathfrak P(S)$, by 
  Lemma \ref{lemma:lyapunov}, there exists a sequence $\{E_n^i:n\ge 1\} \subseteq \Sigma^\ast_{S_i}$
  of coalitions in $\mathscr E^*$ such that for all $n\ge 1$,
  \[
  \mu^\ast(E_n^i\cap T_1^\ast)= \delta \mu^\ast (S_i\cap T_1^\ast)
  \] 
  and
  \[
  \lim_{n\to \infty} \int_{E_n^i}(a^\ast(\cdot,\omega)-h^\ast(\cdot))
  d\mu^\ast= \delta \int_{S_i}(a^\ast(\cdot,\omega)-h^\ast(\cdot))
  d\mu^\ast.
  \]
  Let $E_n= \bigcup\{E_n^i:i\in \mathfrak P(S)\}$ for all $n\geq 1$. Then
  \[
  \lim_{n\to \infty} \int_{E_n}(a^\ast(\cdot,\omega)-h^\ast(\cdot))
  d\mu^\ast= \delta \int_{S}(a^\ast(\cdot,\omega)-h^\ast(\cdot))d\mu^\ast.
  \]
  Pick an $n_0\ge 1$ such that 
  \[
  b:=\int_{E_{n_0}}(a^\ast(\cdot,\omega)-h^\ast (\cdot)) d\mu^\ast\gg 0. 
  \]
  Put $R^\ast= E_{n_0}$ and define an allocation $g^\ast$ in $\mathscr E^\ast(\omega)$ such that
  \[
  g^\ast(t)= h^\ast(t)+ \frac{b}{\mu^\ast(R^\ast)}.
  \]
  Thus, $f^\ast$ is blocked by $R^\ast$ via $g^\ast$ in $\mathscr E^\ast(\omega)$.
  \end{proof}
  
  \subsection{The ex-post core and the the fine core}
  In this subsection, we will present and prove our main result of this section. We assume that our economy
  $\mathscr E$ only admits finitely many information structures. More precisely, we assume that each agent's 
  information partition is a member of $\{\mathscr Q_1, \cdots,\mathscr Q_n\}$. For any $1\le i \le n$ 
  and any coalition $S$, let $S^i= \{t\in S: \Pi_t= \mathscr Q_i\}$ and
  \[
  \mathfrak P(S)= \{i: \mu(S^i)> 0, 1\le i \le n \}.
  \]
  The symbol $\bigvee \{\mathscr Q_i:i\in \mathfrak P(S)\}$ is used to denote the $\sigma$-algebra 
  on $\Omega$, which is generated by the common refinement of members of $\{\mathscr Q_i:i\in 
  \mathfrak P(S)\}$. We will need the following two additional assumptions.
  
  \medskip
  \noindent
  (${\rm A}_5$) $\mu(T^i)> 0$ for all $1\le i \le n$, $T=\bigcup \{T^i:1\le i\le n\}$ and 
  $\bigvee_{i=1}^n \mathscr Q_i = \mathscr F$.
  
  \medskip
  \noindent
  $({\rm A_6})$ All large agents in $\mathscr E$ are of the same type, i.e., having the same characteristics.
  
  \medskip
  Following \cite{Wilson:78}, an \emph{information structure} for a coalition $S$ in an economy $\mathscr 
  E$ is a family $\{\mathscr{G}_t: t\in S\}$ of $\sigma$-algebras on $\Omega$ such that ${\mathscr G}_t
  \subseteq \mathscr F$ for all $t\in S$ and $\{t\in S: \mathscr{G}_t= \mathscr G\} \in \Sigma$ for any
  $\sigma$-algebra ${\mathscr G}$ on $\Omega$ with ${\mathscr G}\subseteq \mathscr F$. A 
  \emph{communication system} for a coalition $S$ is an information structure $\{\mathscr{G}_t:
  t\in S\}$ for $S$ such that
  \[
  \mathscr F_t\subseteq \mathscr{G}_t\subseteq \bigvee \{\mathscr Q_i: i\in {\mathfrak P}(S)\},\ 
  \mbox{ $\mu$-a.e. on $S$.}
  \]
  A communication system  $\{\mathscr{G}_t: t\in S\}$ for a coalition $S$ is called 
  a \emph{full} communication system if $\mathscr{G}_t=\bigvee \{\mathscr Q_i:i\in \mathfrak P(S)\}$, 
  $\mu$-a.e. on $S$. 
  
  \begin{definition}[\cite{Wilson:78}]
  An allocation $f$ in $\mathscr E$ is said to be \emph{fine blocked} by a coalition $S$ in an economy
  $\mathscr E$ if there are an allocation $g$ in $\mathscr E$, a communication system $\{\mathscr{G}_t: 
  t\in S\}$ for $S$ and a nonempty event $\Omega_0\in \bigcap_{t\in S} \mathscr G_t$ such that for all 
  $\omega\in \Omega_0$,
  \[
  \int_{S}g(\cdot, \omega)d\mu= \int_{S}a(\cdot, \omega) d\mu
  \]
  and
  \[
  \mathbb E_t[U(t,\cdot, g(t, \cdot))| \mathscr{G}_t](\omega)> \mathbb E_t[U(t,\cdot, f(t, \cdot))|
  \mathscr{G}_t](\omega),\ \mbox{$\mu$-a.e. on $S$.}
  \]
  The \emph{fine core} of $\mathscr E$, denoted by ${\bf C}^{fine}({\mathscr E})$, is the set of all 
  allocations that cannot be fine blocked by any coalition in $\mathscr E$.
  \end{definition}
  
  Let $\mathscr A\neq \emptyset$. For any allocation $f$ in $\mathscr E$, let ${\bar f}: T
  \times \Omega\to \mathbb R^\ell_+$ be an allocation in $\mathscr E$ defined by
  \[
  {\bar f}(t, \omega)= \left\{
  \begin{array}{ll}
  f(t, \omega), & \mbox{if $(t, \omega)\in T_0\times \Omega$;}\\[0.5em]
  {\displaystyle \frac{1}{\mu(T_1)}\int_{T_1} f(\cdot, \omega)d \mu}, & \mbox{if $(t,\omega)\in T_1
  \times \Omega$.}
  \end{array}
  \right.
  \]
  
  \begin{theorem} \label{lem:fineex} 
  Assume that an economy $\mathscr E$ satisfies \emph{$({\rm A}'_1)$}, \emph{$({\rm A}_2)$-$({\rm 
  A}_3)$}, \emph{$({\rm A}'_4)$} and \emph{$({\rm A}_5)$-$({\rm A}_6)$}. If $f \in {\bf C}^{fine}
  (\mathscr E)$ and $\mathscr A\neq \emptyset$, then 
  \[
  U(t,\omega, {\bar f}(t, \omega)) =U(t,\omega, f(t, \omega)),\ \  \mbox{$\mu$-a.e. on $T$}
  \] 
  and for all $\omega \in \Omega$.
  \end{theorem}
  
  \begin{proof}
  Firstly, we show that for all $(t,\omega) \in T_1 \times \Omega$,
  \[
  U(t,\omega,f(t,\omega))\ge U(t,\omega,{\bar f}(t,\omega)).
  \] 
  Suppose the contrary. There exist a state $\omega_0\in \Omega$ and a coalition $S \subseteq T_1$ 
  such that
  \[
  U(t,\omega_0, \bar{f}(t,\omega_0))>U(t,\omega_0, f(t,\omega_0))
  \]
  for all $t\in S$. For a sequence $\{r_m:m\ge 1\}\subseteq (0, 1)$ converging to 1, the function 
  $\zeta_m: S\to \mathbb R^\ell_+$, defined by
  \[
  \zeta_m(t)= U(t,\omega_0, r_m \bar{f}(t,\omega_0))-U(t,\omega_0, f(t,\omega_0)),
  \]
  is $\Sigma_S$-measurable. For each $m \ge 1$, put 
  \[
  S_m=\{t\in S: \zeta_m(t)>0\}.
  \] 
  As $S= \bigcup_{m\ge 1} S_m$, then $\mu(S_{m_0})> 0$ for some $m_0 \ge 1$. Put
  \[
  z_0 = - \frac{(1- r_{m_0})\mu(S_{m_0})}{\mu(T_1)} \int_{T_1} a(\cdot,\omega_0)d \mu.
  \]
  Choose an $\varepsilon> 0$ with $z_0 + B(0, 2\varepsilon)\subseteq -\mathbb R^\ell_{++}$. For 
  each $i\in\mathfrak P(T_0)$ and $R\in \Sigma_{T_0^i}$, let
  \[
  b_i (R)= \int_{R} (f(\cdot, \omega_0)- a(\cdot, \omega_0)) d\mu-\frac{r_{m_0}\mu(S_{m_0})}{\mu(T_1)} 
  \int_{T_0^i}(f(\cdot, \omega_0)- a(\cdot,\omega_0))d\mu.
  \]
  Applying Lemma \ref{lemma:lyapunov} with $\delta =\frac{r_{m_0}\mu(S_{m_0})}{\mu(T_1)}$, we can get
  a coalition $R_i$ in $\mathscr E$ with $R_i\subseteq T_0^i$ such that
  $b_i(R_i)\in B\left(0,\frac{\varepsilon}{\mathfrak P(T_0)}\right)$.
  Put 
  \[
  R_0=\bigcup\{R_i:i\in \mathfrak P(T_0)\}.
  \] 
  Let $E= R_0\cup S_{m_0}$. Then, by $({\rm A}_5)$ and $({\rm A}_6)$, we have 
  \[
  \bigvee\{\mathscr Q_i: i\in {\mathfrak P}(E)\}=\mathscr F.
  \] 
  Pick an $x\in B(0, \varepsilon)\cap \mathbb R^\ell_{++}$ and define $g: T\to \mathbb R^\ell_{+}$ by
  \[
  g(t)= \left\{
  \begin{array}{ll}
  f(t, \omega_0)+ \frac{x}{\mu(R_0)}, & \mbox{if $t \in R_0$;}\\[0.5em]
  r_{m_0} {\bar f}(t, \omega_0), & \mbox{if $t\in S_{m_0}$;}\\[0.5em]
  f(t, \omega_0), & \mbox{otherwise.}
  \end{array}
  \right.
  \]
  Then, 
  \[
  U(t,\omega_0, g(t))> U(t,\omega_0, f(t,\omega_0)),\ \ \mbox{ $\mu$-a.e. on $E$.}
  \]	
  Furthermore,
  \[
  \int_E g d\mu= \int_{R_0} f(\cdot, \omega_0)d\mu+ \frac{r_{m_0}\mu(S_{m_0})}{\mu(T_1)}
  \int_{T_1} f(\cdot, \omega_0)d\mu+ x.
  \]
  Using the fact that
  \[
  \int_{T_1} (f(\cdot,\omega_0)- a(\cdot, \omega_0))d\mu = -\int_{T_0} (f(\cdot,\omega_0)- a(\cdot, 
  \omega_0))d \mu,
  \]
  we can easily verify that for all $\omega \in \Omega$,
  \[
  -z_0 + \int_E (g(\cdot)- a(\cdot, \omega_0))d\mu=\sum_{i\in \mathfrak P(T_0)}b_i(R_i)+ x\in B(0,
  2\varepsilon).
  \]
  It follows that
  \[
  d = \int_E a(\cdot, \omega_0)- \int_E g d\mu \gg 0.
  \]
  Then, the function $h: E\to \mathbb R^\ell_+$ defined by $h(t)= g(t)+ \frac{d}{\mu(E)}$ for all 
  $t \in T$, satisfies 
  \[
  U(t,\omega_0,h(t))> U(t,\omega_0, f(t, \omega_0)),\ \ \mbox{$\mu$-a.e. on $E$.}
  \] 
  Define 
  \[
  \Omega_0=\bigcap\{\mathscr Q_i(\omega_0): i\in \mathfrak P(E)\}, 
  \]
  where $\mathscr Q_i(\omega_0)$ 
  is the atom in $\mathscr Q_i$ containing $\omega_0$. Note that the set 
  \[
  A_t=\{\omega\in \Omega: U(t,\omega,h(t))> U(t,\omega, f(t, \omega))\}
  \]
  is $\mathscr F$-measurable and $\omega_0\in A_t$ for all $t\in E$. Consequently, 
  by $({\rm A}_5)$ and $({\rm A}_6)$, $\Omega_0\subseteq 
  A_t$ for all $t\in E$. Since the map $\displaystyle \omega \mapsto \int_E a(\cdot,\omega)d\mu$ is 
  $\mathscr F$-measurable, we have 
  \[
  \Omega_0\subseteq \left\{\omega\in \Omega:\int_E hd\mu= \int_E a(\cdot,\omega)d\mu\right\}.
  \] 
  Define another function $y: T \times \Omega\to \mathbb R^\ell_+$ by
  \[
  y(t, \omega)= \left\{
  \begin{array}{ll}
  h(t), & \mbox{if $(t, \omega)\in E\times \Omega_0$;}\\[0.5em]
  a(t, \omega), & \mbox{otherwise.}
  \end{array}
  \right.
  \]
  Note that $y$ is an allocation. Thus, we has
  \[
  \mathbb E_t\left[U(t,\cdot, f(t, \cdot))| \bigvee \{\mathscr Q_i:i\in \mathfrak P(E)\}\right]= 
  U(t,\cdot, f(t, \cdot))
  \]
  and
  \[
  \mathbb E_t\left[U(t,\cdot, y(t, \cdot))|\bigvee \{\mathscr Q_i:i\in \mathfrak P(E)\}\right]=
  U(t, \cdot,y(t, \cdot)).
  \]
  Furthermore, for all $\omega\in \Omega_0$, we have 
  \[
  U(t,\omega, y(t,\omega))> U(t,\omega, f(t, \omega)),\ \ \mbox{$\mu$-a.e. on $E$},
  \] 
  which implies that $f$ is fine blocked by $E$ via $y$. This 
  contradicts with the assumption that $f \in {\bf C}^{fine}(\mathscr E)$. Hence, 
  \[
  U(t,\omega, f(t,\omega))\ge U(t,\omega, {\bar f}(t, \omega))
  \] 
  for all $(t,\omega)\in T_1 \times \Omega$. 
  
  \medskip
  Suppose that there are a	state $\omega_\ast\in \Omega$ and a coalition $D\subseteq T_1$ such that
  \[
  U(t,\omega_*, f(t,\omega_*))> U(t,\omega_*, {\bar f}(t, \omega_*))
  \] 
  for all $t\in D$. By Jensen's inequality,
  \[
  U\left(t, \omega_\ast, \int_D \frac{f(\cdot,\omega_\ast)}{\mu(D)}\
  d\mu\right)> U\left(t, \omega_\ast, {\bar f}(t,\omega_\ast)\right)
  \]
  and
  \[
  U\left(t, \omega_\ast, \int_{T_1\setminus D} \frac{f(\cdot,\omega_\ast)}{\mu(T_1\setminus D)}\ 
  d\mu\right) \ge U\left(t, \omega_\ast, {\bar f}(t, \omega_\ast)\right).
  \]
  Let $\delta= \frac{\mu(D)} {\mu(T_1)}$. Since
  \[
  {\bar f}(t, \omega_*)= \delta \int_D \frac{f(\cdot,\omega_\ast)}{\mu(D)}\ d\mu + (1- \delta) 
  \int_{T_1\setminus D} \frac{f(\cdot,\omega_\ast)}{\mu(T_1\setminus D)}\ d\mu,
  \]
  then 
  \[
  U(t, \omega_\ast, {\bar f}(t, \omega_\ast))> U(t, \omega_\ast, {\bar f}(t, \omega_\ast)). 
  \]
  This is a contradiction, which implies that 
  \[
  U(t, \omega, {\bar f}(t, \omega))= U(t, \omega, f(t,\omega))
  \] 
  for all $(t, \omega)\in T\times \Omega$.
  \end{proof}
  
  \begin{corollary} \label{coro:expost}
  Assume that an economy $\mathscr E$ satisfies \emph{$({\rm A}'_1)$}, \emph{$({\rm A}_2)$-$({\rm 
  A}_3)$}, \emph{$({\rm A}'_4)$} and \emph{$({\rm A}_5)$-$({\rm A}_6)$}. Then $f \in {\bf C}
  ({\mathscr E})$ if and only if $\bar{f} \in {\bf C}({\mathscr E})$.
  \end{corollary}
  
  The following theorem is an extension of \cite[Theorem 3.1]{Einy-Moreno-Shitovitz:00} to a mixed 
  economy.
  
  \begin{theorem} \label{thm:Ex-postCore}
  Assume that an economy $\mathscr E$ satisfies \emph{$({\rm A}'_1)$}, \emph{$({\rm A}_2)$-$({\rm 
  A}_3)$}, \emph{$({\rm A}'_4)$} and \emph{$({\rm A}_5)$-$({\rm A}_6)$}. If either $|\mathscr A|
  \ge 2$ or $\mathscr A = \emptyset$, then ${\bf C}^{fine}(\mathscr E) \subseteq {\bf C}(\mathscr E)$.
  \end{theorem}
  
  \begin{proof}
  First, we assume $|\mathscr A|\ge 2$ and $f \in {\bf C}^{fine}(\mathscr E)$. If $f \not \in {\bf C}
  (\mathscr E)$, then, by Corollary \ref{coro:expost}, ${\bar f} \not \in {\bf C}(\mathscr E)$. By 
  Theorem \ref{thm:mainexpost}, there is an $\omega_0\in \Omega$ such that ${\bar f}(\cdot, \omega_0) 
  \not\in {\bf C}\left(\mathscr E(\omega_0)\right)$. Next, we consider an allocation ${\bar f}^\ast: 
  T^\ast \to \mathbb R^\ell_+$ in ${\mathscr E}^*(\omega_0)$ defined by
  \[
  {\bar f}^*(t)= \left\{
  \begin{array}{ll}
  f(t, \omega_0), & \mbox{if $t\in T_0$;}\\[0.5em]
   {\displaystyle \frac{1}{\mu(T_1)}\int_{T_1} f(\cdot, \omega)d \mu}, & \mbox{if $t\in T_1^*$.}
  \end{array}
  \right.
  \]
  It is clear that ${\bar f}^\ast \not \in {\bf C}\left(\mathscr E^\ast(\omega_0)\right)$. Choose an 
  arbitrary $A_{n_0}\in \mathscr A$ and let $\mu(A_{n_0})= \varepsilon >0$. By Vind's theorem (see
  \cite[Theorem 3.1]{Bhowmik-Cao:paper2} or \cite{Vind:72}), ${\bar f}^\ast$ is blocked by a 
  coalition $S$ in $\mathscr E^\ast(\omega_0)$, which can be chosen such that $\mu^\ast(S)=\mu(T_0) 
  +\varepsilon$, if
  \[
  \mu^\ast(T_1^\ast\setminus A_{n_0}^\ast)< \min \{\mu(T_0^i): i\in \mathfrak P(T_0)\},
  \]
  and  
  \[
  \mu^\ast(S)> \mu^*(T^*) -\min\{\mu(T_0^i): i\in \mathfrak P(T_0)\},
  \]
  otherwise. In either case, it can be checked that $\mu^\ast(S \cap T_1^\ast)\geq \varepsilon$ and 
  $\mathfrak P(S)= \{1,2,\cdots,n\}$. By Lemma \ref{lem:privateblocked}, we can have a coalition $E^*$ 
  in $\mathscr E^\ast(\omega_0)$ with 
  \[
  E^\ast\subseteq\bigcup\{S\cap (T^i)^\ast: i\in \mathfrak P(S)\},
  \] 
  $\mathfrak P(E^\ast)= \mathfrak P(S)$ and $\mu^\ast(E^\ast\cap T_1^\ast) = \varepsilon$, which
  blocks ${\bar f}^\ast$ via $h^\ast$ in $\mathscr E^\ast(\omega_0)$. Consider a coalition $E$ in 
  $\mathscr E$ defined by $E= (E^\ast \cap T_0)\cup A_{n_0}$. Then, $\mathfrak P(E)= \{1,2,\cdots,
  n\}$. Now, we consider a function $h: E\to \mathbb R^\ell_+$ defined by
  \[
  h(t) = \left\{
  \begin{array}{ll}
  h^\ast(t), & \mbox{if $t\in E^\ast\cap T_0$;}\\[0.5em]
  \displaystyle{\frac{1}{\varepsilon}\int_{E^\ast\cap T_1^\ast} h^\ast d\mu^\ast}, & \mbox{otherwise.}
  \end{array}
  \right.
  \]
  Obviously, 
  \[
  U(t, \omega_0, h(t))> U(t, \omega_0, {\bar f}(t, \omega_0)),\ \mbox{$\mu$-a.e. on $E^\ast\cap 
  T_0$.}
  \] 
  By Jensen's inequality, if $t\in A_{n_0}$, we have
  \[
  U(t, \omega_0, h(t))> U(t, \omega_0, {\bar f}(t, \omega_0)). 
  \]
  Moreover,
  \[
  \int_E h d\mu = \int_E a(\cdot,\omega_0)d\mu.
  \]
  Similar to that in Theorem \ref{lem:fineex}, we can define $\Omega_0$ and an allocation $y: T 
  \times \Omega\to \mathbb R^\ell_+$ in $\mathscr E$ such that
  \[
  y(t, \omega) = \left\{
  \begin{array}{ll}
  h(t), & \mbox{if $(t, \omega)\in E\times \Omega_0$;}
  \\[0.5em]
  a(t, \omega), & \mbox{otherwise.}
  \end{array}
  \right.
  \]
  Note that
  \[
  \mathbb E_t\left[U(t,\cdot, f(t, \cdot))| \bigvee \{\mathscr Q_i:i\in \mathfrak P(E)\}\right]= 
  U(t,\cdot, f(t, \cdot))
  \]
  and
  \[
  \mathbb E_t\left[U(t,\cdot, y(t, \cdot))|\bigvee \{\mathscr Q_i:i\in \mathfrak P(E)\}\right]=
  U(t, \cdot,y(t, \cdot)).
  \]
  Thus, $f$ is fine blocked by $E$ via $y$. This is a contradiction.
  	
  \medskip
  In case that $\mathscr A= \emptyset$, $f \in {\bf C}^{fine}(\mathscr E)$ but $f \not \in {\bf C}
  (\mathscr E)$, an argument similar to the previous case can be applied. The major difference is 
  that in this case, the blocking coalition $E$ can be chosen such that
  \[
  \mu(E)> \mu(T_0)-\min\{\mu(T_0^i): i\in \mathfrak P(T_0)\}.
  \]
  The rest part of the proof is almost identical with that of the previous case.
  \end{proof}
  
  Applying the core-Walras equivalence theorem in \cite{Greenberg-Shitovitz:86, Shitovitz:73}, we have 
  the following corollary.
  	
  \begin{corollary} \label{coro:walrasian}
  Assume that an economy $\mathscr E$ satisfies \emph{$({\rm A}'_1)$}, \emph{$({\rm A}_2)$-$({\rm 
  A}_3)$}, \emph{$({\rm A}'_4)$} and \emph{$({\rm A}_5)$-$({\rm A}_6)$}. If $f \in {\bf C}^{fine}
  (\mathscr E)$, then $f(\omega, \cdot)\in {\bf WA}(\omega)$ for every $\omega \in \Omega$.
  \end{corollary}

  \section{Concluding Remarks} \label{sec:conclusion}
  
  A considerable amount of research work on different types of core and equilibrium concepts in economies 
  with asymmetric information can be found in the literature. In particular, attempts in extending the 
  classical equivalence of competitive equilibrium allocations and core allocations in  a standard complete 
  information economy have been made. For instance, the reader can refer to \cite{Einy-Moreno-Shitovitz:00, 
  Einy-Moreno-Shitovitz:01, Forges-Heifetz-Minelli:01, Maus:04}. In this paper, we focus our study on 
  the ex-post core and its relationships to the fine core and the set of rational expectations equilibrium 
  allocations, in two major parts.
  
  \medskip
  The fist part of the paper concerns the relationship between the ex-post core and the set of rational 
  expectations equilibrium allocations. For our economic model, we apply a variety of techniques from 
  Set-Valued Analysis to establish a representation result on the ex-post core (see Theorem 
  \ref{thm:mainexpost}). In an early paper \cite{Bhowmik-Cao:16}, Bhowmik and Cao established a 
  similiar representation result for the set of rational expectations equilibrium allocations. These 
  two representation results imply that for our model of asymmetric information economies, rational 
  expectations equilibrium allocations are contained in the ex-post core (see Corollary \ref{coro:welfare}).
  
  \medskip
  To our knowledge, the idea of representing the ex-post core (resp. the set of rational expectations 
  equilibrium allocations) by selections from the core (resp. competitive equilibrium) correspondence 
  of the associated family of complete information economies is from \cite{Einy-Moreno-Shitovitz:00b}.
  The fundamental difference between \cite{Einy-Moreno-Shitovitz:00b} and this paper is that economies 
  in \cite{Einy-Moreno-Shitovitz:00b} are assumed to have only finitely many states of nature, while 
  economies in this paper are allowed to have infinitely many states of nature. Representation results
  in \cite{Einy-Moreno-Shitovitz:00b}, together with Aumann's Core Equivalence Theorem, imply that if 
  the economy is atomless and the utility function of each trader is measurable with respect to his 
  information field, then the set of rational expectations equilibrium allocations coincides with the 
  ex-post core. However, this generally does not hold, when an atomless asymmetric information economy
  has infinitely many states of nature (see Example \ref{exam:sinclusion}).
  
  \medskip
  The second part of this paper emphasizes on the relationship between the fine core and the ex-post 
  core in oligopolistic economies. We show that under standard assumptions and the assumption that 
  there are only finitely many different information structures and all information is the joint 
  information of agents, the fine core is contained in the ex-post core, if an economy is either 
  atomless or has at least two large agents with the same characteristics (see Theorem 
  \ref{thm:Ex-postCore}). This result can be regarded as an extension of the corresponding result in 
  \cite{Einy-Moreno-Shitovitz:00}, where economies are assumed to be atomless only and have only 
  finitely many states of nature. It would be interesting to know if the conclusion of Theorem 
  \ref{thm:Ex-postCore} still holds for a mixed economy with only one large agent or with two large 
  agents having different characteristics.

%
%


\bibliographystyle{spbasic}      


  \end{document}